\documentclass[a4paper,10pt]{article}

\usepackage[T1]{fontenc}
\usepackage{graphicx}
\usepackage{graphics}
\usepackage{color}
\usepackage{amsmath}
\usepackage{amssymb}
\usepackage{amsfonts}
\usepackage[english]{babel}
\usepackage{enumitem}
\usepackage{amsthm}
\newtheorem{theorem}{Theorem}
\usepackage{mathtools}
\usepackage{fullpage}
\usepackage{array}
\usepackage{authblk}

\usepackage[colorinlistoftodos]{todonotes}

\usepackage{algpseudocode}
\usepackage[nothing]{algorithm}
\algtext*{EndWhile}% Remove "end while" text
\algtext*{EndIf}% Remove "end if" textbf
\algtext*{EndFor}% Remove "end if" textbfu

\floatname{algorithm}{Procedure}

\newcommand{\anum}{k} 
\newcommand{\cnum}{z} 

\newcommand{\colE}{c}         %color of an edge
\newcommand{\colV}{c}         %color of a node (vertex)
\newcommand{\colS}{\tilde{c}} %color of a searcher
\newcommand{\cS}{\mathcal{S}} %search strategy
\newcommand{\cD}{\mathcal{D}} %the set of contaminated edges
\newcommand{\cC}{\mathcal{C}} %clean edges 
\newcommand{\cR}{\mathcal{R}} %recontaminated edges
\newcommand{\card}[1]{\lvert #1\rvert} % the number of elements in #1 
\newcommand{\numofs}[1]{\card{#1}}
\newcommand{\st}{\hspace{0.1cm}\bigl|\bigr.\hspace{0.1cm}}
\newcommand{\SimpTree}{\tilde{T}}   		% simplified tree of color subgraphs

\newcommand{\sn}[1]{\textup{\texttt{s}}(#1)}     % search number
\newcommand{\csn}[1]{\textup{\texttt{cs}}(#1)}   % connected search number
\newcommand{\hsn}[1]{\textup{\texttt{hs}}(#1)}   % heterogenous search number
\newcommand{\hcsn}[1]{\textup{\texttt{hcs}}(#1)} % heterogenous connected search number

\newcommand{\msn}[1]{\textup{\texttt{ms}}(#1)}
\newcommand{\mcsn}[1]{\textup{\texttt{mcs}}(#1)}
\newcommand{\mhsn}[1]{\textup{\texttt{mhs}}(#1)}
\newcommand{\mhcsn}[1]{\textup{\texttt{mhcs}}(#1)}

\newcommand{\YES}{\textup{\texttt{YES}}}

\newcommand{\problemHGS}{\textup{HGS}}
\newcommand{\problemSAT}{3\textup{-SAT}}
\newcommand{\problemHCGS}{\textup{HCGS}}
\newcommand{\agentsInS}[3]{#1_{#3}(#2)}  % f-cja kolorow, kolor, strategia

\newcommand{\first}[1]{R\textup{-}pr(v_{#1})}
\newcommand{\areas}[1]{Areas(#1)}
\newcommand{\progress}[1]{#1\textup{-}progress}
\newcommand{\reconf}[1]{#1\textup{-}reconfig}

\newcommand{\cB}{\mathcal{B}}

\newcommand{\cF}{\mathcal{F}}
\newcommand{\cG}{\mathcal{G}}
\newcommand{\cK}{\mathcal{K}} 
 
\newcommand{\cX}{\mathcal{X}} 
\newcommand{\cY}{\mathcal{Y}} 

\newcommand{\f}[2]{R\textup{-}pr(\tilde{v}_{#1}, #2)}%{s_{#1,#2}}
\newcommand{\induced}[2]{#1[#2]}
\newcommand{\cset}[1]{Q(#1)} % set of colors of a (sub) graph
\newcommand{\clean}[2]{Clean(#1,#2)}
\newcommand{\cont}[2]{Cont(#1,#2)}

\newcommand{\att}[2]{Attempt(#1,#2)}

\newcommand{\atts}[3]{#3\textup{-}Attempt(#1,#2)}
\newcommand{\cT}{\mathcal{T}} 
 
\newcommand{\cQ}{\mathcal{Q}} % set of colors

\newcommand{\lb}[1]{\beta(#1)}
\newcommand{\problemMHGS}{\textup{MHGS}}

\newcommand{\F}{S\textup{-}Attempt(\tilde{P},1)}
\newcommand{\Fb}{\atts{\tilde{P}_{\tilde{b}}^{+}} {i(0)}{S}}

\newcommand{\literal}[2]{l_{#1, #2}}

\newcommand{\NPCtrue}[1]{\textup{\texttt{T}}_{#1}} %{teal_{#1}}
\newcommand{\NPCfalse}[1]{\textup{\texttt{F}}_{#1}} %{fuchsia_{#1}}
\newcommand{\NPCclause}[1]{\textup{\texttt{C}}_{#1}} %{cyan_{#1}}
\newcommand{\NPCvariable}[1]{\textup{\texttt{V}}_{#1}}  %{violet_{#1}}
\newcommand{\RED}{\textup{\texttt{R}}}
\newcommand{\NPCvalve}[2]{\textup{\texttt{O}}_{#1, #2}}

\newcommand{\TSAT}{T_{\textup{SAT}}}
\newcommand{\TSATP}{\tilde{T}_{\textup{SAT}}}
\newcommand{\ProcClean}{Clean}

\newtheorem{lemma}{Lemma}[section]

\newtheorem{observation}{Observation}[section]

\newtheorem{fact}{Fact}[section]

\algblockdefx{Foreach}{EndForeach}[1]{\textbf{Foreach} #1}{\textbf{EndForeach}}

\title{Searching by Heterogeneous Agents\thanks{Research partially supported by National Science Centre (Poland) grant number 2015/17/B/ST6/01887.}\thanks{A preliminary version of this paper appeared in the Proc. 11th International Conference on Algorithms and Complexity (CIAC 2019).}}
\author[1]{Dariusz Dereniowski}
\author[2]{{\L}ukasz Kuszner}
\author[1]{Robert Ostrowski}

\affil[1]{Faculty of Electronics, Telecommunications and Informatics, Gda{\'n}sk University of~Technology, Poland\\
\texttt{deren@eti.pg.edu.pl}, \texttt{robostro@student.pg.gda.pl}}
\affil[2]{Faculty of Mathematics, Physics and Informatics, University of Gda{\'n}sk, Poland\\
\texttt{lkuszner@inf.ug.edu.pl}}

\date{}

\begin{document}
\maketitle

\begin{abstract}
In this work we introduce and study a pursuit-evasion game in which the search is performed by heterogeneous entities. We incorporate heterogeneity into the classical edge search problem by considering edge-labeled graphs: once a search strategy initially assigns labels to the searchers, each searcher can be only present on an edge of its own label. We prove that this problem is not monotone even for trees and we give instances in which the number of recontamination events is asymptotically quadratic in the tree size. Other negative results regard the NP-completeness of the monotone, and NP-hardness of an arbitrary (i.e., non-monotone) heterogeneous search in trees. These properties show that this problem behaves very differently from the classical edge search. On the other hand, if all edges of a particular label form a (connected) subtree of the input tree, then we show that optimal heterogeneous search strategy can be computed efficiently.
\end{abstract}


\sloppy

\section{Introduction} \label{sec:intro}

Consider a scenario in which a team of searchers should propose a search strategy, i.e., a sequence of their moves, that results in capturing a fast and invisible fugitive hiding in a graph.
This strategy should succeed regardless of the actions of the fugitive and the fugitive is considered captured when at some point it shares the same location with a searcher. In a strategy, the searchers may perform the following moves: a searcher may be placed/removed on/from a vertex of the graph, and a searcher may slide along an edge from currently occupied vertex to its neighbor.
The fugitive may represent an entity that does not want to be captured but may as well be an entity that wants to be found but is constantly moving and the searchers cannot make any assumptions on its behavior.
There are numerous models of graph searching that have been introduced and studied and these models can be produced by enforcing some properties of the fugitive (e.g., visibility, speed, randomness of its movements), properties of the searchers (e.g., speed, type of knowledge provided as an input or during the search, restricted movements, radius of capture), types of graphs (e.g., simple, directed) or by considering different optimization criteria (e.g., number of searchers, search cost, search time).

One of the central concepts in graph searching theory is \emph{monotonicity}.
Informally speaking, if a search strategy has the property that once a searcher traversed an edge (and by this action it has been verified that in this very moment the fugitive is not present on this edge) it is guaranteed (by the future actions of the searchers) that the edge remains inaccessible to the fugitive, then we say that the search strategy is \emph{monotone}.
In most graph searching models, including the edge search recalled above, it is not beneficial to consider search strategies that are not monotone. Such a property is crucial for two main reasons:
firstly, knowing that monotone strategies include optimal ones reduces the algorithmic search space when finding good strategies and secondly, monotonicity places the problem in the class NP.

To the best of our knowledge, all searching problems studied to date are considering the searchers to have the same characteristics.
More precisely, the searchers may have different `identities' which allows them to differentiate their actions but their properties like speed, radius of capture or interactions with the fugitive are identical.
However, there exist pursuit-evasion games in which some additional device (like a sensor or a trap) is used by the searchers \cite{ClarkeC06,ClarkeN00,ClarkeN01,SundaramKC16}.
In this work we introduce a searching problem in which searchers are different: each searcher has access only to some part of the graph.
More precisely, there are several \emph{types} of searchers, and for each edge $e$ in the graph, only one type of searchers can slide along $e$.
We motivate this type of search twofold.
First, referring to some applications of graph searching problems in the field of robotics, one can imagine scenarios in which the robots that should physically move around the environment to execute a search strategy may not be all the same. Thus some robots, for various reasons, may not have access to the entire search space.
Our second motivation is an attempt to understand the concept of monotonicity in graph searching.
In general, the graph searching theory lacks of tools for analyzing search strategies that are not monotone, where a famous example is the question whether the connected search problem belongs to NP \cite{BFFFNST12}. (In a connected search we require that at each point of the strategy the subgraph that is guaranteed not to contain the fugitive is connected; for a formal definition see Section~\ref{sec:formulation}.) In the latter, the simplest examples that show that recontamination may be beneficial for some graphs are quite complicated~\cite{YDA09}.
The variant of searching that we introduce has an interesting property: it is possible to construct relatively simple examples of graphs in which multiple recontaminations are required to search the graph with the minimum number of searchers.
Moreover, it is interesting that this property holds even for trees.

\subsection{Related work}
In this work we adopt two models of graph searching to our purposes.
Those models are the classical \emph{edge search} \cite{Parsons76,Petrov82}, which is historically the first model studied, and its connected variant introduced in~\cite{BFFS02}.
As an optimization criterion we consider minimization of the number of searchers a strategy uses.

The edge search problem is known to be monotone \cite{monotonicity_in_graph_searching,LaPaugh93} but the connected search is not~\cite{YDA09}.
See also~\cite{FominT03} for a more unified approach for proving monotonicity for particular graph searching problems.
Knowing that the connected search is not monotone, a natural question is what is the `price of monotonicity', i.e., what is the ratio of the minimum number of searchers required in a monotone strategy and an arbitrary (possibly non-monotone) one?
It follows that this ratio is a constant that tends to $2$~\cite{Dereniowski12SIDMA}.
We remark that if the searchers do not know the graph in advance and need to learn its structure during execution of their search strategy then this ratio is $\Omega(n/\log n)$ even for trees~\cite{INS09}.
An example of recently introduced model of \emph{exclusive graph searching} shows that internal edge search with additional restriction that at most one searcher can occupy a vertex behaves very differently than edge search.
Namely, considerably more searchers are required for trees and exclusive graph searching is not monotone even in trees~\cite{BlinBN17,MarkouNP17}.
Few other searching problems are known not to be monotone and we only provide references for further readings~\cite{deren_ipl09,FN08,YDA09}.
Also see~\cite{FominHT04} for a searching problem for which determining whether monotonicity holds turns out to be a challenging open problem. 

Since we focus on trees in this work, we briefly survey a few known results for this class of graphs.
An edge search strategy that minimizes the number of searchers can be computed in linear time for trees~\cite{MegiddoHGJP88}.
Connected search is monotone and can be computed efficiently for trees~\cite{BFFFNST12} as well.
However, if one considers weighted trees (the weight of a vertex or edge indicate how many searchers are required to clean or prevent recontamination), then the problem turns out to be strongly NP-complete, both for edge search~\cite{MT09} and connected search~\cite{Dereniowski11}.
On the other hand, due to \cite{Dereniowski12,Dereniowski12SIDMA} both of these weighted problems have constant factor approximations.
The class of trees usually turns out to be a very natural subclass to study for many graph searching problems --- for some recent algorithmic and complexity examples see e.g. \cite{AminiCN15,DereniowskiDTY15,DyerYY08,GMNP10}.
See also~\cite{HollingerKS10} for an approximation algorithm for general graphs that performs by adopting optimal search strategies computed for spanning trees of the input graph.

We conclude by pointing to few works that use heterogeneous agents for solving different computational tasks, mostly in the area of mobile agent computing.
These include modeling traffic flow~\cite{SeanQian2017183}, meeting~\cite{LunaFSVY17} or rendezvous~\cite{DereniowskiKKK15,FarrugiaGKP15,FeinermanKKR14}.
We also note that heterogeneity can be introduced by providing weights to mobile agents, where the meaning of the weight is specific to a particular problem to be solved~\cite{BartschiC0D0HP17,CzyzowiczGKK11,KawamuraK15}, while in~\cite{Czyzowicz2014} authors consider patrolling by robots with distinct speeds and visibility ranges.

\subsection{Our work --- a short outline}

We focus on studying monotonicity and computational complexity of our heterogeneous graph searching problem that we formally define in Section~\ref{sec:formulation}.
We start by proving that the problem is not monotone in the class of trees (Section~\ref{sec:monotonicity}).
Then in Section~\ref{sec:hard} we show  that, also in trees, monotone search with heterogeneous searchers is NP-complete. In Section~\ref{sec:nmhard} we prove that the general, non-monotone, searching problem is NP-hard for trees.

Our investigations suggest that the essence of the problem difficulty is hidden in the properties of the availability areas of the searchers.
For example, the problem becomes hard for trees if such areas are allowed to be disconnected. 
To formally argue that this is the case we give, in Section~\ref{sec:easy}, a polynomial-time algorithm that finds an optimal search strategy for heterogeneous searchers in case when each color class induces a connected subtree.
This result holds also for the connected version of the heterogeneous graph search problem.

Section~\ref{sec:preliminaries} is concluded with Table~\ref{tab:summary} that points out the complexity and monotonicity differences between the classical and connected edge search with respect to our problem.

\section{Preliminaries} \label{sec:preliminaries}

In this work we consider simple edge-labeled graphs $G=(V(G),E(G),\colE)$, i.e., without loops or multiple edges, where $\colE\colon E(G)\to \{1,\ldots,\cnum\}$ is a function that assigns labels, called \emph{colors}, to the edges of $G$.
Then, if $\colE(\{u,v\})=i$, $\{u,v\}\in E(G)$, then we also say that vertices $u$ and $v$ \emph{have} color $i$.
Note that vertices may have multiple colors, so by  $\colV(v) := \{ \colE(\{u,v\}) \, : \, \{u,v\} \in E(G)\}$ we will refer to the set of colors of a vertex $v\in V(G)$.

\subsection{Problem formulation} \label{sec:formulation}

We will start by recalling the classical \emph{edge search} problem~\cite{Parsons76} and then we will formally introduce our adaptation of this problem to the case of heterogeneous searchers.

An (edge) \emph{search strategy} $\cS$ for a simple graph $G=(V(G),E(G))$ is a sequence of moves $\cS=(m_1,\ldots,m_{\ell})$.
Each move $m_i$ is one of the following actions:
\begin{enumerate} [label={\normalfont{(M\arabic*)}},leftmargin=*]
\item\label{move:placing} placing a searcher on a vertex,
\item\label{move:removing} removing a searcher from a vertex,
\item\label{move:sliding} sliding a searcher present on a vertex $u$ along an edge $\{u,v\}$ of $G$, which results in a searcher ending up on $v$.
\end{enumerate}
We often write for brevity `move $i$' in place of `move $m_i$'.

Furthermore, we recursively define for each $i\in\{0,\ldots,\ell\}$ a set $\cC_i$ such that $\cC_i$, $i>0$, is the set of edges that are \emph{clean} after the move $m_i$ and $\cC_0$ is the set of edges that are clean prior to the first move of $\cS$.
Initially, we set $\cC_0=\emptyset$.
For $i>0$ we compute $\cC_i$ in two steps.
In the first step, let $\cC_i'=\cC_{i-1}$ for moves \ref{move:placing} and \ref{move:removing}, and let $\cC_{i}'=\cC_{i-1}\cup\{\{u,v\}\}$ for a move \ref{move:sliding}.
In the second step compute $\cR_i$ to consists of all edges $e$ in $\cC_i'$ such that there \emph{exists} a path $P$ in $G$ such that no vertex of $P$ is occupied by a searcher at the end of move $m_i$, one endpoint of $P$ belongs to $e$ and the other endpoint of $P$ belong to an edge not in $\cC_{i-1}$.\footnote{We point out that another way of computing the set $\cC_i$ is possible. Namely, start again with the same set $\cC_i'$. Then, check if the following condition holds: there exists an edge $e$ in $\cC_i'$ that is adjacent to an edge not in $\cC_i'$ and their common vertex is not occupied by a searcher. In such case, remove $e$ from $\cC_i'$. Keep repeating such an edge removal from $\cC_i'$ until there is no such edge $e$. Then, set $\cC_i=\cC_i'$ and $\cR_i=E(G)\setminus\cC_i'$.}
We stress out that it is enough that only one such path exists, and in particular, if a contaminated edge is adjacent to a clean edge $e$, then $e$ becomes contaminated when their common vertex $v$ is not occupied by a searcher. In such case, $P$ consists of the vertex $v$ only.
Then, set $\cC_i=\cC_{i}'\setminus\cR_i$.
If $\cR_i\neq\emptyset$, then we say that the edges in $\cR_i$ become \emph{recontaminated} (or that \emph{recontamination occurs in $\cS$} if it is not important which edges are involved). If $l_e$ is the number of times the edge $e$ becomes recontaminated during a search strategy, then the value $\sum_{e\in E(G)}l_e$ is referred to as the number of \emph{unit recontaminations}.
Finally, we define $\cD_i=E(G)\setminus\cC_i$ to be the set of edges that are \emph{contaminated} at the end of move $m_i$, $i>0$, where again $\cD_0$ refers to the state prior to the first move.
Note that $\cD_0=E(G)$.
We require from a search strategy that $\cC_{\ell}=E(G)$.

Denote by $V(m_i)$ the vertices occupied by searchers at the end of move $m_i$.
We write $\numofs{\cS}$ to denote the number of searchers used by $\cS$ understood as the minimum number $\anum$ such that at most $\anum$ searchers are present on the graph in each move.
Then, the \emph{search number of} $G$ is
\[\sn{G}=\min\left\{\numofs{\cS}\st \textup{$\cS$ is a search strategy for $G$}\right\}.\]

If the graph induced by edges in $\cC_i$ is connected for each $i\in\{1,\ldots,\ell\}$, then we say that $\cS$ is \emph{connected}.
We then recall the \emph{connected search number of} $G$:
\[\csn{G}=\min\left\{\numofs{\cS}\st \textup{$\cS$ is a connected search strategy for $G$}\right\}.\]

\medskip
We now adopt the above classical graph searching definitions to the searching problem we study in this work.
For an edge-labeled graph $G=(V(G),E(G),\colE)$, a search strategy assigns to each of the $\anum$ searchers used by a search strategy a color: the color of searcher $j$ is denoted by $\colS(j)$. This is done prior to any move, and the assignment remains fixed for the rest of the strategy.
Then again, a search strategy $\cS$ is a sequence of moves with the following constraints: in move~\ref{move:placing} that places a searcher $j$ on a vertex $v$ it holds $\colS(j)\in\colV(v)$; move~\ref{move:removing} has no additional constraints; in move~\ref{move:sliding} that uses a searcher $j$ for sliding along an edge $\{u,v\}$ it holds $\colS(j)=\colE(\{u,v\})$.
Note that, in other words, the above constraints enforce the strategy to obey the requirement that at any given time a searcher may be present on a vertex of the same color and a searcher may only slide along an edge of the same color.
To stress out that a search strategy uses searchers with color assignment $\colS$, we refer to as a \emph{search $\colS$-strategy}.
We write $\agentsInS{\colS}{j}{\cS}$ to refer to the number of searchers with color $j$ in a search strategy $\cS$.

Then we introduce the corresponding graph parameters $\hsn{G}$ and $\hcsn{G}$ called the \emph{heterogeneous search number} and \emph{heterogeneous connected search number} of $G$, where $\hsn{G}$ (respectively $\hcsn{G}$) is the minimum integer $\anum$ such that there exists a (connected) search  $\colS$-strategy for $G$ that uses $\anum$ searchers.

Whenever we write $\sn{G}$ or $\csn{G}$ for an edge-labeled graph $G=(V,E,\colE)$ we refer to $\sn{G'}$ and $\csn{G'}$, respectively, where $G'=(V,E)$ is isomorphic to $G$.

\medskip
We say that a search strategy $\cS$ is \emph{monotone} if no recontamination occurs in $\cS$.
Analogously, for the search numbers given above, we define \emph{monotone}, \emph{connected monotone}, \emph{heterogeneous monotone} and \emph{connected heterogeneous monotone} search numbers denoted by $\msn{G}$, $\mcsn{G}$, $\mhsn{G}$ and $\mhcsn{G}$, respectively, to be the minimum number of searchers required by an appropriate monotone search strategy.

\medskip
The decision versions of the combinatorial problems we study in this work are as follows:
\begin{description}
\item[\emph{Heterogeneous Graph Searching} Problem] (\problemHGS) \\
Given an edge-labeled graph $G=(V(G),E(G),\colE)$ and an integer $\anum$, does it hold $\hsn{G}\leq \anum$?
\item[\emph{Heterogeneous Connected Graph Searching} Problem] ($\problemHCGS$) \\
Given an edge-labeled graph $G=(V(G),E(G),\colE)$ and an integer $\anum$, does it hold $\hcsn{G}\leq \anum$?
\end{description}
In the optimization versions of both problems an edge-labeled graph $G$ is given as an input and the goal is to find the minimum integer $\anum$, a labeling $\colS$ of $\anum$ searchers and a (connected) search $\colS$-strategy for $G$.

\begin{table}[ht!]
\label{tab:summary}
\centering
\begin{tabular}{ >{\centering\arraybackslash}m{1.5cm}||>{\centering\arraybackslash}m{3cm}|>{\centering\arraybackslash}m{3.5cm}|>{\centering\arraybackslash}m{3.5cm}|  }

  & Monotone  &Non-monotone & Complexity\\
 \hline
  \hline
  edge search   & arbitrary graphs  \cite{Parsons76,Petrov82,MT09,MegiddoHGJP88} & & P for trees \cite{MegiddoHGJP88}, NPC for weighted trees \cite{MT09}, NPC for arbitrary graphs \cite{MegiddoHGJP88}\\
  \hline
 connected edge search & trees  \cite{BFFFNST12, BFFS02} & arbitrary graphs \cite{YDA09} & P for trees \cite{BFFFNST12}, NPC for weighted trees \cite{Dereniowski11}, NPH for arbitrary graphs \cite{BFFFNST12}  \\
  \hline
 \problemHGS    & & trees [Theorem \ref{non_mon_theo}] & NPH for trees [Theorem \ref{nph_t_theorem}] \\

\hline
\end{tabular}
\caption{Monotonicity and complexity summary of our problems in comparison with the classical and connected edge search problems for trees and arbitrary graphs.}
 \end{table}

\subsection{Additional notation and remarks}

For some nodes $v$ in $V(G)$ we have $\card{\colV(v)}>1$, such connecting nodes we will call \emph{junctions}.
Thus a node $v$ is a junction if there exist two edges with different colors incident to $v$.

We define an \emph{area} in $G$ to be a maximal subgraph $H$ of $G$ such that for every two edges $e, f$ of $H$, there exists a path $P$ in $H$ connecting an endpoint of $e$ with and endpoint of $f$ such that $P$ contains no junctions. We further extend our notation to denote by $c(H)$ the color of all edges in area $H$. Note that two areas of the same color may share a junction. Let $\areas{G}$ denote all areas of $G$.
Two areas are said to be \emph{adjacent} if they include the same junction.

\begin{fact}\label{fact:leaves}
If $T$ is a tree and $v$ is a junction that belongs to some area $H$ in $T$, then $v$ is a~leaf (its degree is one) in $H$.
\qed
\end{fact}

\begin{fact}
If $T$ is a tree, then any two different areas in $T$ have at most one common node which is a junction.
\qed
\end{fact}

\begin{lemma}
\label{lem:clean1}
Given a~tree $T = (V(T), E(T), \colE)$ and any area $H$ in $T$, any search $\colS$\nobreakdash-strategy for $T$ uses at least $\sn{H}$ searchers of color $c(H)$.
\end{lemma}
\begin{proof}
If there are less than $\sn{H}$ searchers of color $c(H)$, then the area $H$ can not be cleaned, as searchers of other colors can only be placed on leafs of $H$.
\end{proof}

We now use the above lemma to obtain a lower bound for the heterogeneous search number of a graph $G=(V(G), E(G), c:E(G)\rightarrow\{1,\ldots, \cnum\})$.
Define
\[\lb{G}=\sum_{i=1}^{\cnum}{ \max \left\{\sn{H} \, \st \, H\in \areas{G}, c(H)=i\right\}}.\]

Using Lemma~\ref{lem:clean1} for each area we obtain the following: 
\begin{lemma} \label{lem:lower}
For each tree $T$ it holds $\hsn{T} \geq \lb{T}$.
\qed
\end{lemma}

\section{Lack of monotonicity} \label{sec:monotonicity}

Restricting available strategies to monotone ones can lead to increase of  heterogeneous search number, even in case of trees.
We express this statement in form of the following main theorem of this section:

\begin {theorem}
\label{non_mon_theo}
There exists a tree T such that $\mhsn{T}>\hsn{T}$.
\end{theorem}

In order to prove this theorem we provide an example of a tree $T_{l}=(V,E,c)$, where $l\geq 3$ is an integer, which cannot be cleaned with $\lb{T_{l}}$ searchers using a monotone search strategy, but there exists a non-monotone strategy, provided below, which achieves this goal.
Our construction is shown in Figure~\ref{fig:monotonicity}.

We first define three building blocks needed to obtain $T_{l}$, namely subtrees $T'_{1}$, $T'_{2}$ and $T''_{l}$.
We use three colors, i.e., $\anum\geq 3$.
The construction of the tree $T'_{i}, i\in\{1,2\}$, starts with a root vertex $q_{i}$, which has 3 further children connected by edges of color $1$. Each child of $q_{i}$ has 3 children connected by edges of color $2$.

For the tree $T_{l}'', l\geq 3$, take vertices $v_{0},\ldots,v_{l+1}$ that form a path with edges $e_{x}=\{v_{x},v_{x+1}\}$, $x\in \{0,\ldots,l\}$.
We set $\colE(e_{x})=x\mod 3+1$.
We attach one pendant edge with color $x\mod 3 + 1$ and one with color $(x-1)\mod 3+1$ to each vertex $v_{x}, x\in\{1,\ldots,l\}$.
Next, we take a path $P$ with four edges in which two internal edges are of color 2 %red
and two remaining edges are of color 3. %blue. 
To finish the construction of $T_{l}''$, identify the middle vertex of $P$, incident to the two %red 
edges of color 2, with the vertex $v_{0}$ of the previously constructed subgraph.

We link two copies of $T'_{i}$, $i\in\{1,2\}$, by identifying two endpoints of the path $P$ with the roots $q_{1}$ and $q_2$ of $T_1'$ and $T_2'$, respectively, obtaining the final tree $T_{l}$ shown in Fig.~\ref{fig:monotonicity}.

\begin{figure}[!ht]
\begin{center}
\includegraphics[scale=0.7]{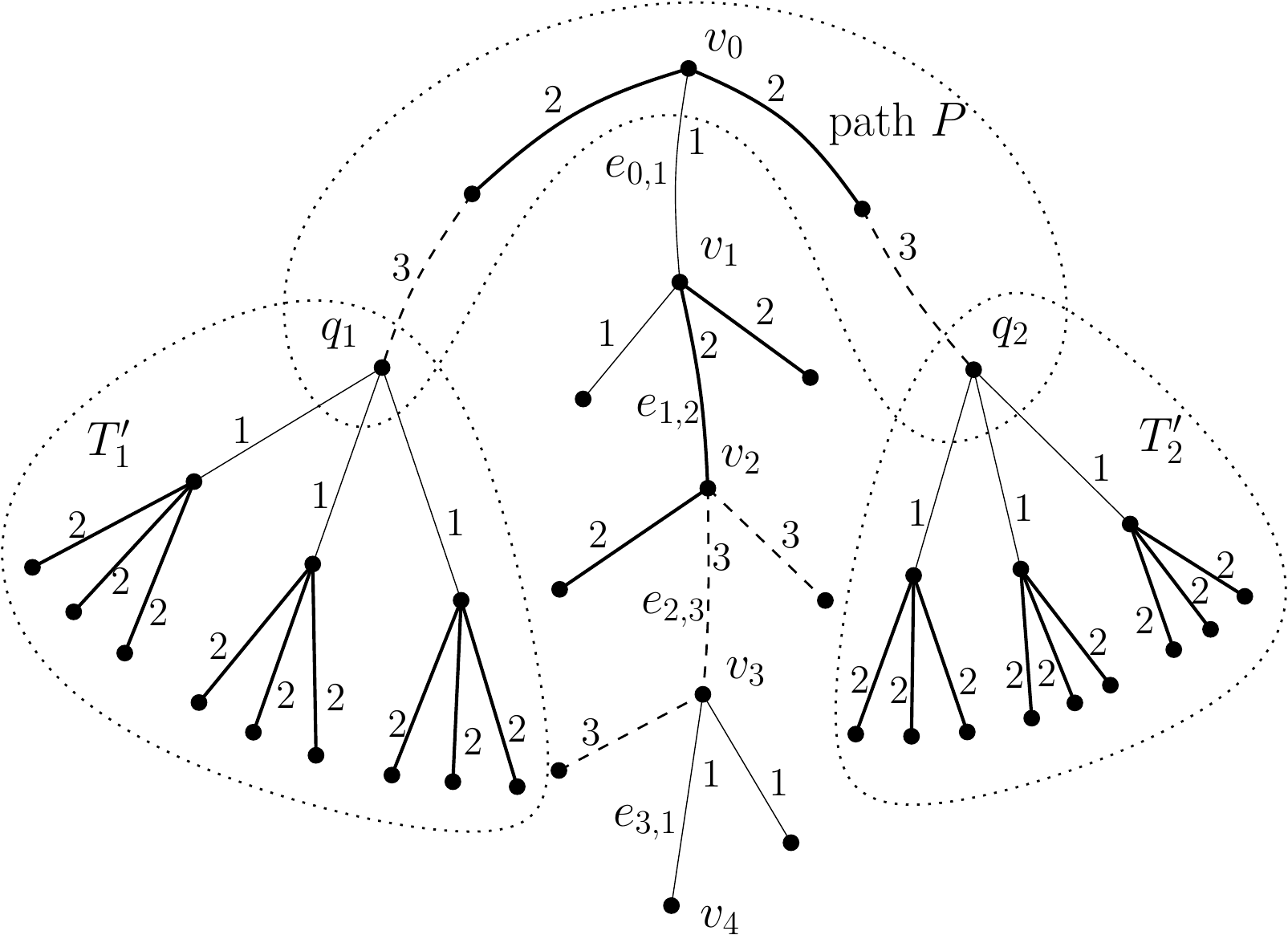}
\caption{The construction of $T_{3}$ ($l=3$) from the trees $T_1'$, $T_2'$ and $T_{3}''$. Regular, heavy  and dashed edges have labels $1,2$ and $3$, respectively.
}
\label{fig:monotonicity}
\end{center}
\end{figure}

Now, we are going to analyze a potential monotone search $\colS$-strategy $\cS$ using $\lb{T_l}=3$ searchers.
Thus, by Lemma~\ref{lem:clean1}, $\cS$ uses one searcher of each color.
We define a notion of a \emph{step} for $\cS=(m_1,\ldots,m_l)$ to refer to some particular moves of this strategy.
We distinguish the following steps that will be used in the lemmas below:
\begin{enumerate}
\item step $t_{i}, i\in\{1,2\}$, equals the minimum index $j$ such that at the end of move $m_j$ all searchers are placed on the vertices of $T_i'$ (informally, this is the first move in which all searchers are present in $T'_{i}$); 
\item step $t'_{i}, i\in\{1,2\}$, is the maximum index $j$ such that at the end of move $m_j$ all searchers are placed on the vertices of $T_i'$ (informally, this is the last move in which all searchers are present in $T'_{i}$); 
\item steps $t_{3},t'_{3}$ are, respectively, the minimum and maximum indices $j$ such that at the end of move $m_j$ all searchers are placed on the vertices in $V(P)\cup V(T''_{l}).$
\end{enumerate}

We skip a simple proof that all above steps are well defined, i.e., for any search strategy using $3$ searchers for $T$ each of the steps $t_i,t_i'$, $i\in\{1,2,3\}$, must occur (for trees $T_1'$ and $T_2'$ this immediately follows from $\sn{T'_{i}}=3$ for $i \in\{1,2\}$).

\begin{lemma}
\label{lemma:sequence_mon}
For each monotone $\tilde{c}$-search strategy $\cS$ for $T_3$ it holds: $t_{1}\leq t_{1}'<t_{3}\leq t'_{3}<t_{2}\leq t'_{2}$ or  $t_2 \leq t_{2}'<t_{3}\leq t'_{3}<t_{1}\leq t'_{1} $.
\end{lemma}
\begin{proof}
Intuitively, we prove the lemma using the following argument: in the process of cleaning $T_i'$, $i\in\{1,2\}$, all three searchers are required for some steps, and therefore a monotone strategy could not have partially cleaned $T_{3-i}'$ or $T_3''$ prior to this point.

The arguments used to prove this lemma do not use colors, so atomic statements about search strategies for subgraphs can be analyzed using simple and well known results for edge search model. 
Furthermore, due to the symmetry of $T$, it is enough to analyze only the case when $t_{1}<t_{2}$. Note that $t_{i}\leq t'_{i}, i\in\{1,2,3\}$, follows directly from the definition.  
The vertices $q_{i}$, $i\in\{1,2\}$, have to be guarded at some point between move $t_{i}$ and move $t_{i}'$ because $\sn{T'_{i}}=3$.
Because each step $t_{j}$, $j\in\{1,2,3\}$, uses all searchers, it cannot be performed if a searcher preventing recontamination is required to stay outside of subtree related to the respective step.
The subtrees $T'_{1}$ and $T'_{2}$ contain no common vertices, so $t_{1}<t_{2}$ implies $t_{2}>t_{1}'$, as stated in the lemma.

Suppose for a contradiction that $t_{3}<t_{i}$ for each $i\in\{1,2\}$.
In move $t_3$, since neither of moves $t_{i}'$ has occurred, both subtrees $T'_{1},T_{2}'$ contain contaminated edges.
Moreover, some of the contaminated edges are incident to vertices $q_{i}$.
Thus, any edge of $T_l''$ that is clean becomes recontaminated in the step $\min\{t_1,t_{2}\}$.
Therefore, $t_{1}<t_{3}$ as required.

Now we prove that $t_1'<t_3$.
Suppose for a contradiction that  $t_{1}<t_{3}<t'_{1}$. Consider the move of index $t'_{1}$.
By $t_{3}<t'_{1}$, $T_{l}''$ contains clean edges. By $t_{1}'<t_{2}$, $q_2$ is incident to contaminated edges in $T_{2}'$. Thus, there is a searcher outside of $T_{1}'$ which prevents recontamination of clean edges in $T_{l}''$.
Contradiction with the definition of $t'_{1}$.

In move $t_{2}$ there are no spare searchers left to  guard any contaminated area outside $T'_{2}$ which bypasses  $q_{2}$ and could threaten recontamination of $T_{1}'$, so all edges, including the ones in $T_{l}''$, between those two trees should have been clean already. Therefore step $t'_{3}$ has to have already occurred, which allows us to conclude $t'_{3}<t_{2}$.
\end{proof}

Due to the symmetry of $T_l$, we consider further only the case $t_{1}\leq t_{1}'<t_{3}\leq t'_{3}<t_{2}\leq t_{2}'$.
\begin{lemma}
\label{lemma:on_path}
During each move of index $t\in[t_{1}', t_{2}]$ there is a searcher on a vertex of $P$.
\end{lemma}
\begin{proof}
By $t\geq t_{1}'$, $q_1$ is incident to some clean edges of $T_1'$. By $t\leq t_{2}$, $q_2$ is incident to some contaminated edges from $T_2'$. Hence there has to be a searcher on $q_1$, $q_2$ or a vertex of the path $P$ between them to prevent recontamination. 
\end{proof}

Let $f_i, i\in \{1,\ldots,l-1\}$, be the index of a move such that one of the edges incident to $v_{i}$ is clean, one of the edges incident to $v_{i}$ 
is being cleaned and and all other edges incident to $v_{i}$ are contaminated. 

Notice that $\sn{T_{l}''}=2$, and therefore an arbitrary search strategy $\cS'$ using two searchers to clean a subtree without colors that is isomorphic to $T_{l}''$ follows one of these patterns: 
either the first searcher is placed, in some move of $\cS'$, on $v_1$ and throughout the search strategy it moves  from $v_1$ to $v_{l-1}$ or the first searcher starts at $v_{l-1}$ and moves from $v_{l-1}$ to $v_1$ while $\cS'$ proceeds.
If for each $i\in\{1,\ldots,l-1\}$ the edge $\{v_{i-1},v_{i}\}$ becomes clean prior to the edge $\{v_{i},v_{i+1}\}$ --- we say that such $\cS'$ \emph{cleans $T_{l}''$ from $v_1$  to $v_{l-1}$} and if the edge $\{v_{i-1},v_i\}$ becomes clean after $\{v_i,v_{i+1}\}$ --- we say that such $\cS$ \emph{cleans $T_{l}''$ from $v_{l-1}$ to $v_1$}. 

\begin{lemma}
\label{lemma:third}
Each move of index $f_i,  i\in \{1,\ldots,l-1\}$,  is well defined.
Either $f_{1}<f_{2}<\ldots<f_{l-2}<f_{l-1}$ or  $f_{l-1}<f_{l-2}<\ldots<f_{2}<f_{1}$.
\end{lemma}
\begin{proof}
Consider a move of index $f$ which belongs to $[t_{3}, t_{3}']$ in a search strategy $\cS$.
By Lemma~\ref{lemma:sequence_mon} and Lemma~\ref{lemma:on_path}, a searcher is present on a vertex of $P$ in the move of index $f$. 
Hence, only two searchers can be in $T_{l}''$ in the move
$f$, so $\cS$ cleans $T_{l}''$ from $v_1$ to $v_{l-1}$ or cleans $T_{l}''$ from $v_{l-1}$ to $v_1$.
Note that during an execution of such a strategy there occur moves which satisfy the definition of $f_i$, and therefore there exists well defined $f_i$. %\in [t_{3}, t_{3}']$. 
When $\cS$ cleans $T_{l}''$ from $v_0$ to $v_l$, then $f_{1}<f_{2}<\ldots<f_{l-2}<f_{l-1}$ is satisfied and when $\cS$ cleans $T_{l}''$ from $v_l$ to $v_0$, then $f_{l-1}<f_{l-2}<\ldots<f_{2}<f_{1}$ is satisfied.
\end{proof}

\begin{lemma}
\label{cant_monotone}
There exists no monotone search $\colS$-strategy that uses $3$ searchers to clean $T_l$ when $l\geq7$.
\end{lemma}
\begin{proof}
We use the following intuition in the proof: whenever a search strategy tries to clean the path composed of the vertices $v_0,\ldots,v_{l+1}$, together with the corresponding incident edges, then it periodically needs searchers of all three colors on this path.
While doing this, different vertices of the path $P$ need to be guarded.
More precisely, when the search moves along the former path, it needs to move along $P$ as well.
Due to the fact that $l$ is large enough, the path $P$ is not long enough to avoid recontamination.

The vertex $v_i, i \in \{1,\ldots, l-1\}$, is incident to edges of colors $i\mod 3+1$ and $(i-1)\mod 3+1$, and therefore each move $f_{i}$ uses both searchers of colors $i\mod 3+1$ and $(i-1)\mod 3+1$. By Lemma~\ref{lemma:on_path}, the third searcher, which is of color $(i-2)\mod 3+1$, stays on $P$.

Consider a sequence $f_{6}<f_{5}<\ldots<f_{2}<f_{1}$. Note that it implies that $T_{3}''$ is cleaned from $v_{l-1}$ to $v_1$. Let us show that it is impossible to place a searcher on the vertices of $P$ such that no recontamination occurs in each $f_{i}, i\in \{1,\ldots, 6\}$.   

Consider the move of index $f_{6}$, where searchers of colors 1 and 3 are in $T''_{l}$ and 2 is on $P$. 
Before  move $f_{6}$ an edge incident to $v_{6}$ is clean (by definition of $f_{6}$). No edge incident to  $v_1$ is clean and, by Lemma~\ref{lemma:sequence_mon}, $T_{j}'$ has a clean edge, $j\in \{1,2\}$. In order to prevent recontamination of $T_{j}'$, the searcher is present on $P$, particularly on a vertex of the path from $q_{j}$ to $v_{0}$. It cannot be the vertex $q_{j}$, because $2\notin c(q_{j})$, so the edge of color 3 incident to $q_{j}$ is clean, and the searcher is on one of the remaining two vertices.
Consider the move of index $f_{5}$, in which the searcher of color $1$ is on a vertex $v$ of $P$. The vertex between $q_{j}$ and $v_{0}$ cannot be occupied, due to its colors, and occupying $q_{j}$ would cause recontamination --- only the vertex $v_{0}$ is available, $v=v_0$. 
Consider the move of index $f_{4}$. The vertex $v_{0}$ cannot be occupied, due to its colors. The edge $e_{0}$ cannot be clean before $e_{4}$ is clean, because $T_{l}''$ is cleaned from $v_{l-1}$ to $v_1$. 
Therefore, the searcher on $v_{0}$ cannot be moved towards $q_2$. Monotone strategy fails.

The argument is analogical for a sequence $f_{1}<f_{2}<\ldots<f_{5}<f_{6}$. By Lemma~\ref{lemma:third}, $T_{l}''$ is cleaned either from $v_1$ to $v_{l-1}$ or the other way, which implies that considering the two above cases completes the proof.
\end{proof}

\begin {lemma}\label{monotonicity:non-monotone}
There exists a non-monotone $\colS$-strategy \(\cS\) that cleans $T_l$ using three searchers for each $l\geq 3$.
\end{lemma}
\begin{proof}
The strategy we describe will use one searcher for each of the three colors.
The strategy first cleans the subtree $T'_{1}$ (we skip an easy description how this can be done) and finishes by cleaning the path connecting $q_0$ with $v_0$. Denote the vertex on the path from $v_0$ to $q_2$ as $v$.

Now we describe how the strategy cleans $T_{l}''$ from $v_{1}$ to $v_{l-1}$.
For each $i\in\{1,\ldots,l\}$, the vertex $v_i$ is incident to edges of colors $i\mod 3+1$ and $(i-1)\mod 3+1$ therefore each move $f_{i}$ uses both searchers of colors $i\mod 3+1$ and $(i-1)\mod 3+1$. By Lemma~\ref{lemma:on_path}, the third searcher which is of color $(i-2)\mod 3+1$, stays on $P$.
Informally, while progressing along $T_{l}''$, the strategy makes recontaminations within the path $P$.

We will define $\progress{j}$, $j\in\{1, \ldots, l-1\}$, as a sequence of consecutive moves which clean edges of colors in $\colV(v_{j})$ in $T_{l}''$ and contains the move of index $f_{j}$. Similarly, we introduce $\reconf{i}(u)$, $i\in\{1,2,3\}$, as a minimal sequence of consecutive moves, such that there is a searcher on some vertex $u$ of $P$ in the first move of $\reconf{i}(u)$ and the searcher of color $i$ is present on $P$ in the last move of $\reconf{i}(u)$. Let $u_b$ be the occupied vertex of $P$ after the last move of $b$-th $\reconf{i}(u)$ in $\cS$. Additionally let  $u_0=v_0$. Clean $T_l''$
by iterating for each $j\in \{1,\ldots, l-1\}$ (in this order) the following: $\reconf{a}(u_{j-1})$, followed by $\progress{j}$, where $a=(j-2)\mod 3+1$.

Because determining moves in $\progress{j}$ is straightforward, as they correspond to those in monotone $\colS$-strategy when $f_{1}<f_{2}<\ldots<f_{l-2}<f_{l-1}$, we focus on describing $\reconf{i}(u)$ for each $i\in\{1,2,3\}$. 
$\reconf{1}(u_0)$  consists of a sliding move from $v_0$ to $v$ and a move which places the searcher of color $3$ on $u_1=v$.
$\reconf{2}(u_1)$  consists of a sliding move from $v$ to $v_0$, which causes recontamination, and a move which places the searcher of color $1$ on $u_2=v_0$.
$\reconf{3}(u_2)$  does not contain any sliding moves and places the searcher of color $2$ on $u_3=v_0$.
Because $a=(j-2)\mod 3+1$ and $u_{j-1}=u_{j+2}$
$\reconf{a}(u_{j})$ is identical to $\reconf{a+3}(u_{j+3})$, thus we can describe a strategy which cleans $T_{l}''$ for any given $l$.

When $T_{l}''$ is clean, the vertex $v_0$ is connected to a clean edge and the remaining edges of path $P$ can be searched without further recontaminations. The strategy cleans subtree $T'_{2}$ in the same way as a monotone one.
 
Note that the proposed strategy requires new recontamination whenever a sequence of $f_{i}$  of length 3 repeats itself.
Thus, this $\colS$-strategy cleaning $T_l$ has $\Omega(l)$ unit recontaminations.
Note that the size of the tree $T_l$ is $\Theta(l)$.
\end{proof}

Lemma~\ref{monotonicity:non-monotone} provides a non-monotone search $\colS$-strategy which succeeds with fewer searcher than it is possible for a monotone one, as shown in lemma~\ref{cant_monotone}, which proves Theorem~\ref{non_mon_theo}.

\begin{theorem}
There exist trees such that each search $\colS$-strategy that uses the minimum number of searchers has $\Omega (n^{2})$ unit recontaminations.
\end{theorem}
\begin{proof}
As a proof we use a tree $H_l$ obtained through a modification of the tree $T_l$. In order to construct  $H_l$, we replace each edge on the path $P$ with a path $P_{m}$ containing $m$ vertices, where each edge between them is in the same color as the replaced edge in $T_l$. Clearly  $\hsn{H_l}=\hsn{T_l}$. Note that we can adjust the number of vertices in  $T''_{l}$ and $P_{m}$ of $H_l$ independently of each other.  While the total number of vertices is $n=\Theta(m+l)$, we take $m=\Theta(n)$, $l=\Theta(n)$ in $H_l$.

In order to clean $H_l$, we employ the strategy provided in theorem \ref{monotonicity:non-monotone} adjusted in such a way, that any sliding moves performed on edges of $P$ are replaced by $O(m)$ sliding moves on the corresponding paths of $P_{m}$. As shown previously, the number of times an edge of $P$ in $T_l$, or path $P_{m}$ in $H_l$, which contains $\Theta(m)$ elements, has to be recontaminated  depends linearly on  size of $T''_{l}$.  In the later case the $\colS$-strategy cleaning $H_l$ has $\Omega (ml)= \Omega (n^{2})$ unit recontaminations .
\end{proof}

\section{NP-hardness for trees} \label{sec:hard}

We show that the decision problem $\problemHGS$ is NP-complete for trees if we restrict available strategies to monotone ones.
Formally, we prove that the following problem is NP-complete:
\begin{description}
\item[\emph{Monotone Heterogeneous Graph Searching} Problem] ($\problemMHGS$) \\
Given an edge-labeled graph $G=(V(G),E(G),\colE)$ and an integer $\anum$, does it hold $\mhsn{G}\leq \anum$?
\end{description}
Thus, the rest of this section is devoted to a proof of the following theorem.
\begin{theorem}
The problem $\problemMHGS$ is NP-complete in the class of trees.
\end{theorem}

In order to prove the theorem, we conduct a polynomial-time reduction from Boolean Satisfiability Problem where each clause is limited to at most three literals ($\problemSAT$).
The input to $\problemSAT$ consists of $n$ \emph{variables} $x_{1},\ldots, x_{n}$ and a Boolean formula $C=C_{1}\land C_{2}\land\cdots\land C_{m}$, with each \emph{clause} of the form $C_{i}=(\literal{i}{1} \lor \literal{i}{2} \lor \literal{i}{3})$, where the \emph{literal} $\literal{i}{j}$ is a variable $x_{p}$ or its negation, $\overline{x_{p}}$, $p \in \{1,\ldots,n\}$.
The answer to decision problem is $\YES$ if and only if there exist an assignment of Boolean values to the variables $x_1,\ldots, x_n$ such that the formula $C$ is satisfied.

Given an input to $\problemSAT$, we construct a tree $\TSAT$ which can be searched monotonously by the specified number of searchers if and only if the answer to $\problemSAT$ is $\YES$.
We start by introducing the colors and, informally speaking, we associate them with respective parts of the input:
\begin{itemize}
\item color $\NPCvariable{p}$, $ p\in \{1,\ldots,n\}$, represents the variable $x_{p}$,
\item color $\NPCfalse{p}$ (respectively $\NPCtrue{p}$), $ p\in \{1,\ldots,n\}$, is used to express the fact that to $x_{p}$ may be assigned the Boolean value \emph{false} (\emph{true}, respectively),
\item color $\NPCclause{d}$, $d\in \{1,\ldots,m\}$, is associated with the clause $C_{d}$.
\end{itemize}
We will also use an additional color to which we refer as $\RED$.
 
We denote the set of all above colors by $\cQ$. Note that $\left\vert{\cQ}\right\vert =3n+m+1$.
In our reduction we set $\anum=\left |\cQ  \right |+1+m$ to be the number of searchers.

The construction of the tree starts with a path $P$ of color $\RED$ consisting of $l=4n+3m+4+1$ vertices $v_{i}, i\in\{1,2, \ldots,l \}$. We add 2 pendant edges of color $\RED$ to both $v_1$ and $v_l$.
Define a subgraph $H_{z}$ (see Figure~\ref{fig:npc-ab}(a)) for each color  $z \in \cQ \setminus \{\RED\}$: take a star of color $\RED$ with three edges, attach an edge $e$ of color $z$ to a leaf of the star and then attach an edge $e'$ of color $\RED$ to $e$, so that the degree of each endpoint of $e$ is two.
\begin{figure}[!ht]
\centering
\includegraphics[scale=0.75]{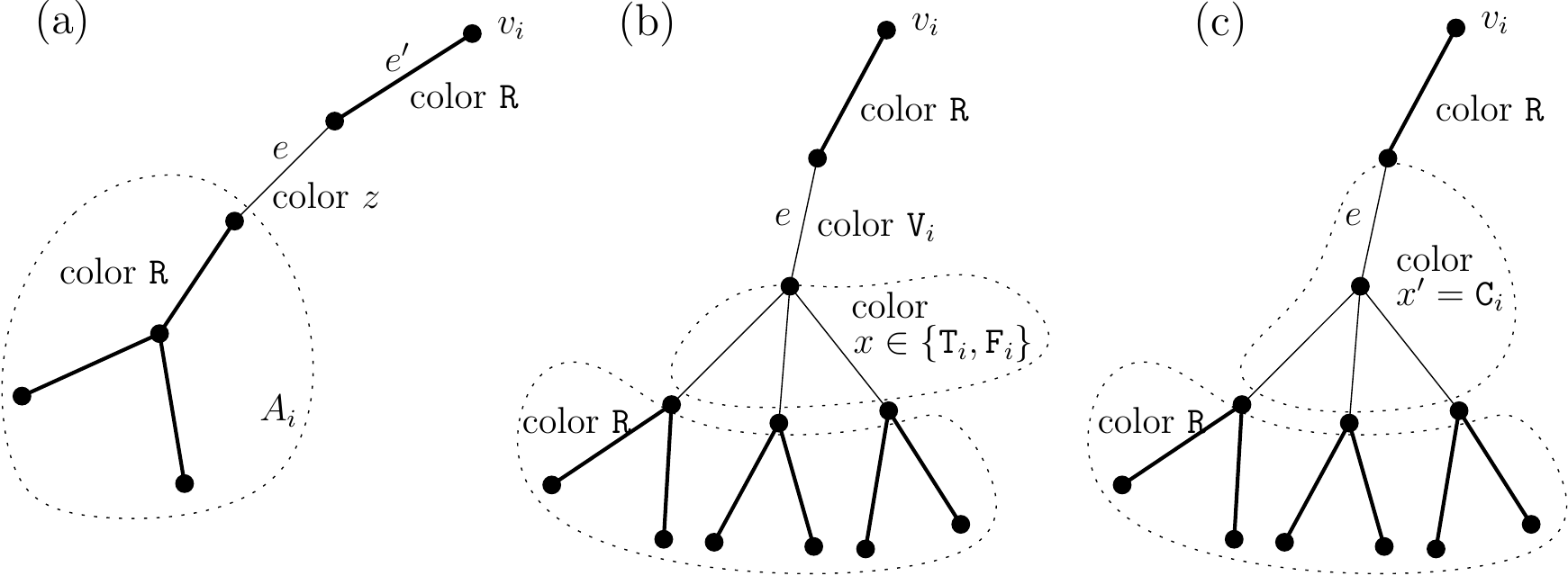}
\caption{Construction of $T$: (a) the subgraph $H_z$; (b) the subgraph $L_{x}$; (c) the subgraph $L'_{x'}$}
\label{fig:npc-ab}
\end{figure}
For each $z\in \cQ\setminus(\{\RED\}\cup\{\NPCclause{1},\ldots, \NPCclause{m}\}) $  take a subgraph $H_{z}$ and join it with $P$ in such a way that the endpoint of $e'$ of degree one in $H_{z}$ is identified with a different vertex in $\{v_{2},\ldots, v_{a}\}$, $a=3n+2m+1$.
For each $z\in \{\NPCclause{1},\ldots, \NPCclause{m}\}$   take two copies of $H_{z}$ and identify each  endpoint of $e'$ of degree one in $H_{z}$ with a different vertex in $\{v_{2},\ldots, v_{a}\}$, which has no endpoint of $e'$ attached to it yet.
The above attachments of the subgraphs $H_z$ are performed in such a way that the degree of $v_{i}$ is three for each $i \in\{2,\ldots,a\}$ (see Figure~\ref{fig:npc-d}). We note that, except for the requirement that no two subgraphs $H_z$ are attached to the same $v_i$, there is no restriction as to which $H_z$ is attached to which $v_i$. The star of color $\RED$ in the subgraph $H_{z}$ attached to the vertex $v_{i}$ is denoted by $A_{i}$. 

For each color $x$ in $X=\{\NPCtrue{i}, \NPCfalse{i}\st i\in\{1,\ldots,n\}\}$ we define a subtree $L_{x}$ (see Figure~\ref{fig:npc-ab}(b)). We start with a root having a single child and an edge of color $\RED$ between them. Then we add an edge $e$ of color $\NPCvariable{i}$ to this child, where $i$ is selected so that it matches $x$ which is either $\NPCtrue{i}$ or $\NPCfalse{i}$. The leaf of $e$ has three further children attached by edges of color $x$. We finish by attaching 2 edges of color $\RED$ to each of the three previous children. 
For each color $x'\in X'=\{\NPCclause{1}, \ldots, \NPCclause{m}\} $  construct a subtree $L'_{x'}$ (see Figure~\ref{fig:npc-ab}(c)) in the same shape but colored in a different way. The edges of color different than $\RED$ in the construction of $L_{x}$ are replaced by edges of color $x'$. We draw attention to the fact that $L'_{x'}$ contains an area of color $x'$ that is a star with four edges. 
We attach to the path $P$ five copies of subtree $L_{x}$ for each  $x\in X$ and five copies of $L'_{x'}$ for each $x'\in X'$ by unifying their roots with the vertex $v_{a+1}$ of $P$ (see Figure~\ref{fig:npc-d}).

We attach  five further  copies
 of $L'_{x'}$ for each $x'\in X'$ by unifying their roots with the vertex $v_{l-1}$ of $P$.
\begin{figure}[!ht]
\centering
\includegraphics[scale=0.7]{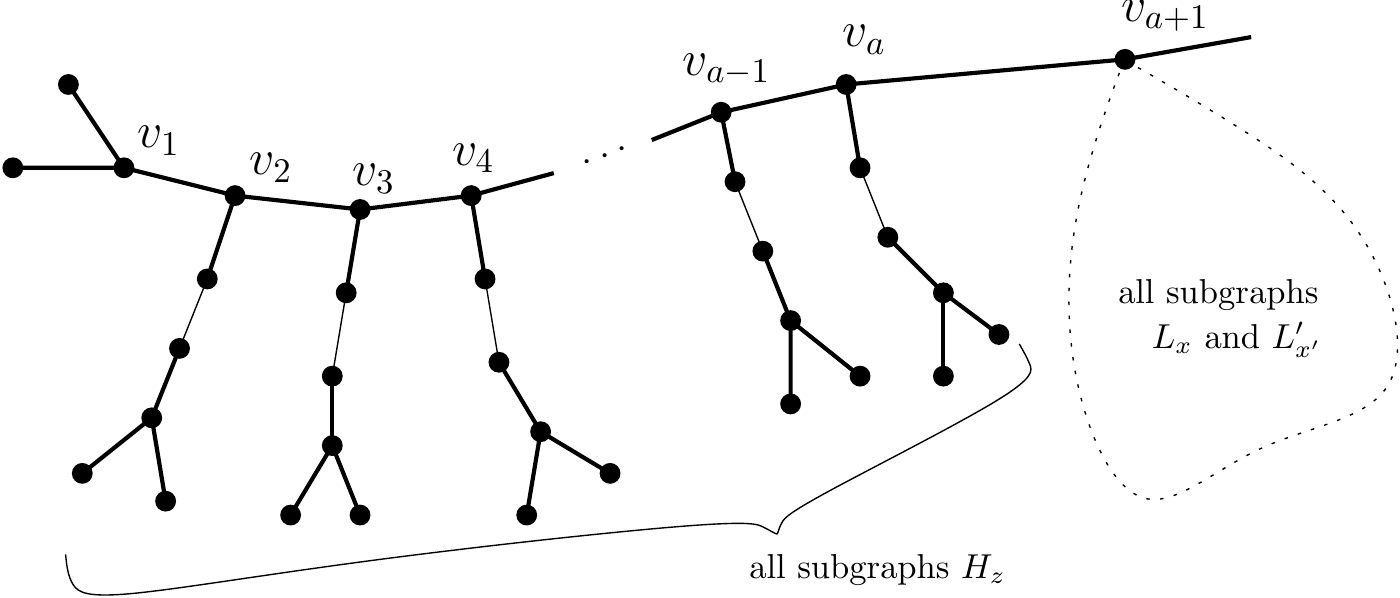}
\caption{Construction of $T$: attachment of subgraphs $H_z$ and the subgraphs $L_x$ and $L'_{x'}$ to the path $P$}
\label{fig:npc-d}
\end{figure}

For each variable $x_{p}$ we construct two subtrees, $S_{p}$ and $S_{-p}$, in the following fashion (see Figure~\ref{fig:npc-cd}(a)): take a star of color $\RED$ with three edges and  attach an endpoint of a path with four edges to a leaf in this star of color $\RED$; the consecutive colors of the path, starting from the endpoint at the star of color $\RED$ are: $\NPCvariable{p}$, $\RED$, $\NPCfalse{p}$, $\RED$ in $S_{-p}$ and $\NPCvariable{p}$, $\RED$, $\NPCtrue{p}$, $\RED$ in $S_{p}$.
For each subtree $S_{p}$ and $S_{-p}$, $p\in\{1,\ldots, n\}$ attach the endpoint of its path of degree one to $v_{a+1+p}$ .
\begin{figure}[!ht]
\centering
\includegraphics[scale=0.75]{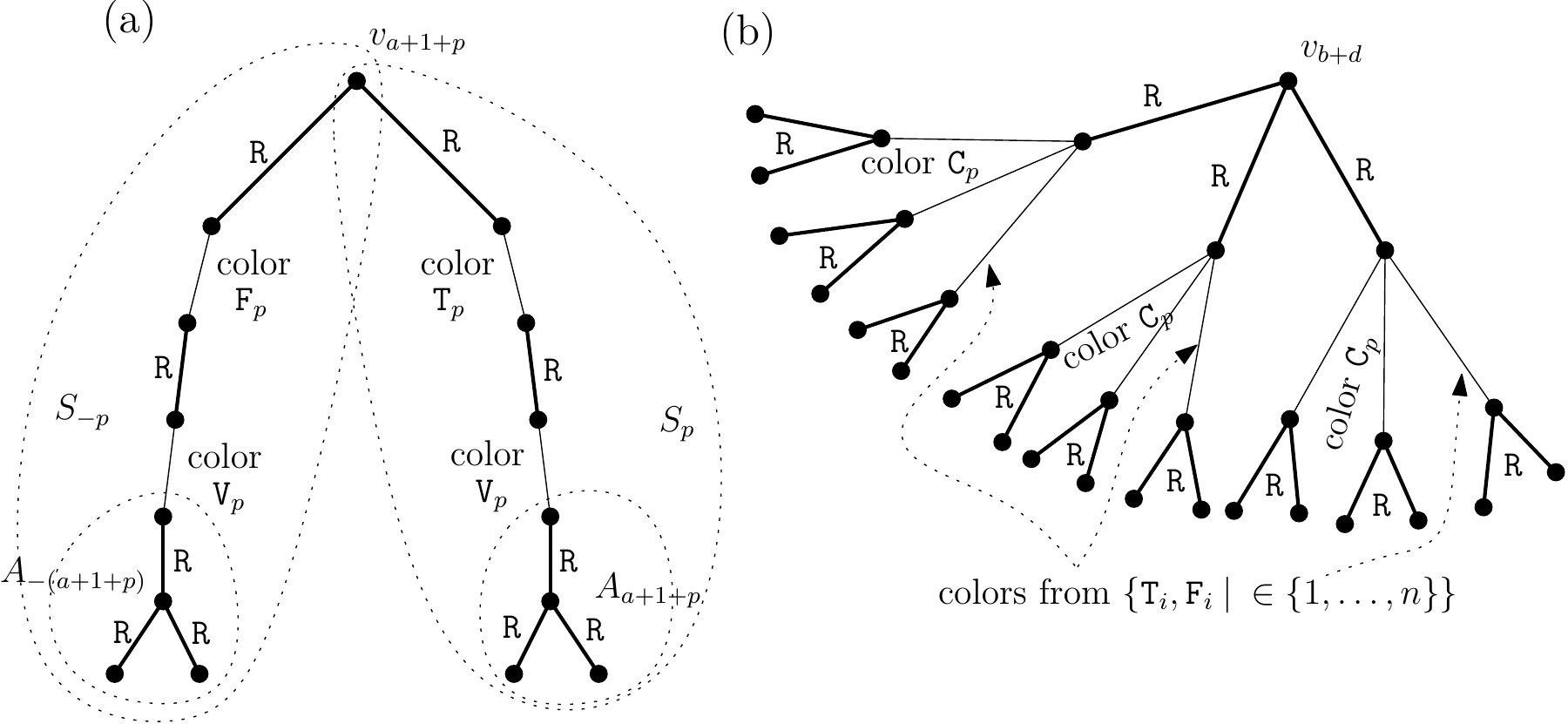}
\caption{Construction of $T$: (a) the variable component constructed from $S_{-p}$ and $S_p$; (b) the clause component that corresponds to $C_d$.}
\label{fig:npc-cd}
\end{figure}
The star of color $\RED$ in  $S_{p}$ attached to $v_{i}$ is denoted by $A_{i}$ and the one in $S_{-p}$ by $A_{-i}$. 

For each clause  $C_{d}, d\in\{1,\ldots, m\}$ we attach three subtrees $L_{d, j}, j\in \{1,2,3\}$, to the vertex $v_{b+d}$, where $b=4n+2m+3 $, one for each literal $\literal{d}{j}$ (see Figure~\ref{fig:npc-cd}(b)). Note that the maximal value of $b+d$ is $l-2$.  We construct $L_{d, j}$ by taking an edge $e$ of color $\RED$ and adding three edges to its endpoint: two of color $\NPCclause{i}$, and one either of color $\NPCtrue{p}$ if $\literal{d}{j}=x_{p}$ or of color $\NPCfalse{p}$ if $\literal{d}{j}=\overline{x_{p}}$. Add two children by the edge of color $\RED$ to each of these three edges. Then attach the endpoint of degree one of the edge  $e$ in $L_{d,j}$ to $v_{b+d}$. We attach a single edge of color $\RED$ to $v_b$. The tree obtained through this construction will be denoted by $\TSAT$. 

The area of color $\RED$ which contains the path $P$ is denoted by $A_{0}$. 
Notice that all areas $A_{i}$ of color $\RED$ have search number two, $\sn{A_{i}}=2$.
For a search strategy for $\TSAT$, we denote the index of the first move in which all searchers  of color $\RED$ are in the area $A_{i}$ as \emph{step} $t_{i}$ and the index of the last such move as \emph{step} $t'_{i}$, $i \in I= \{ 2,\ldots,a\}  \cup \{a+2,\ldots,b-1 \} \cup \{-(b-1),\ldots,-(a+2) \}\cup \{0\}$. Let $R=\{ a+2,\ldots,b-1 \} \cup \{-(b-1),\ldots,-(a+2) \}$ and $L=I\setminus\{R \cup \{0\}\}=\{ 2,\ldots,a\}$ be the two sets which cover all indices of areas $A_{a}: a\in I\setminus \{0\}$. Note that by definition $a+1\not\in L$ and $a+1\not\in R$, and the path from $v_{1}$ to $v_{a+1}$ contains no vertex $v_{j}, j\in R$. Similarly, the path from $v_{a+1}$ to  $v_{l}$ contains no vertex $v_{i}, i \in L$. Informally, we divide the indices in $I\setminus\{0\}$ into two sets: $L$ to the left of  $v_{a+1}$  and $R$ to the right. 

\begin{lemma}[Color assignment]\label{lemma:colors}
A search $\colS$-strategy using $\anum=3n+2m+2$ searchers has to color them in the following fashion: one searcher for each color in $\{\NPCtrue{p},\NPCfalse{p},\NPCvariable{p}\st p\in\{1,\ldots,n\}\}$ and two searchers for each color in $\{\RED\}\cup\{\NPCclause{1}, \ldots, \NPCclause{m}\}$.
\end{lemma}
\begin{proof}
We first compute the lower bound $\lb{\TSAT}$.
By Lemma~\ref{lem:clean1}, at least $3n$ searchers take colors 
 $\NPCfalse{p}$, $\NPCtrue{p}$ and $\NPCvariable{p}, p\in\{1,\ldots,n\}$.
Recall that $L'_{x'}$ contains as a subgraph an area $T'$ of color $x'\in\{\NPCclause{1}, \ldots, \NPCclause{m}\}$ that is a star with three edges and hence $\sn{T'}=2$.
Since there are $m$ such subtrees $L'_{x'}$, $2m$ searchers receive colors $\NPCclause{1}, \ldots, \NPCclause{m}$.
The last two searchers have to be of color $\RED$ in order to clean areas $A_{i}, i\in I$.
Thus, we have shown that $\lb{\TSAT}\geq 3n+2m+2$ and this lower bound is met by the assignment of colors to searchers, as indicated in the lemma.
Using Lemma~\ref{lem:lower} we complete the proof.
\end{proof}

\begin{lemma}\label{lemma:strategy}
Let $x_{1},\ldots, x_{n}$ and a Boolean formula $C=C_{1}\land C_{2}\ldots\land C_{m}$ be an input to $\problemSAT$.
If the answer to $\problemSAT$ is $\YES$, then there exists a search $\tilde{c}$-strategy using $2+3n+2m$ searchers for $\TSAT$.
\end{lemma}
\begin{proof}
We first note the main point as to how a Boolean assignment provides the corresponding search strategy.
Whether a variable is true or false, this dictates which searcher, either of color $\NPCfalse{p}$ or $\NPCtrue{p}$, is placed in the corresponding variable component.
The vertices occupied by these searchers form a separator that disconnects $A_0$ from areas $A_i$, $i\in R$.
The strategy cleans first the latter areas that are protected from recontamination.
As a result, all $A_i$, $i\in R$, become clean.
Then, $A_0$ is cleaned with the clause components along the way: here the fact that the initial Boolean assignment was satisfied ensures that searchers of appropriate colors are available to clean the subsequent clause components.

Suppose that a Boolean assignment to the variables satisfies $C$. The strategy is described as a sequence of instructions.
\begin{enumerate}
\item\label{s:step1} We start by placing a searcher of color dictated by the Boolean assignment in each variable component.
For each $S_{p}$ (respectively $S_{-p}$), place a searcher of color $\NPCtrue{p}$ (respectively $\NPCfalse{p}$) on the vertex that is incident to the edge of color $\NPCtrue{p}$ and does not belong to $A_{0}$ if $x_{p}$ is \emph{false} (respectively \emph{true}).
\item \label{s:step2}  Then, clean $A_{a+1+p}$ (respectively $A_{-(a+1+p)}$) and then the edges in $S_{p}$ (respectively $S_{-p}$) that connect this star of color $\RED$ to the vertex guarded by the searcher of color $\NPCtrue{p}$ (respectively $\NPCfalse{p}$).
Note that this cleaning uses two searchers  of color $\RED$ and the searcher of color $\NPCvariable{p}$.
Then, place the searcher of color $\NPCvariable{p}$ on the vertex that belongs to the edge of color $\NPCvariable{p}$ in $S_{-p}$ (respectively $S_{p}$) and area $A_{-(a+1+p)}$ (respectively $A_{a+1+p}$).
Clean $A_{-(a+1+p)}$ (respectively $A_{a+1+p}$).

By repeating the above for each index $p$, we in particular obtain that all areas $A_i$, $i\in R
$ are clean. Note that if $x_p=false$ (respectively $x_p=true$), then the searcher of color $\NPCtrue{p}$ (respectively $\NPCfalse{p}$) stays in $S_{p}$ (respectively $S_{-p}$) and the searcher of color $\NPCfalse{p}$  (respectively $\NPCtrue{p}$) is available. Let $X$ denote the colors of available searchers among those in colors $\NPCfalse{p}$ and $\NPCtrue{p}$. 

\item\label{s:step3} Then we start cleaning  $A_{0}$ from the vertex $v_{l}$ and move towards $v_{b}$. Clean copies of $L''_{\NPCclause{d}}, d\in\{1,\ldots, m\}$ attached to  $v_{l-1}$.
Consider each approached vertex $v_{h}, h \in \{b+1,b+2,\ldots, l-2\}$, and its three subtrees $L_{d, i}, i\in\{1,2,3\},d \in \{1,\ldots,m\}$ separately. We denote the color different than $\RED$ and $\NPCclause{d}$ in $L_{d, i}$ as $x_{d, i}$. Because $C$ is satisfied, at least one of the literals in each clause is true and for each  $d \in \{1,\ldots,m\}$  there always exists $L_{d, i}$ such that $x_{d, i}$ matches the color of the searcher not assigned to neither $S_{p}$ nor $S_{-p}$, i.e., $x_{d, i}\in X$. Clean each such $L_{d, i}$ by using searchers of color $\NPCclause{d}$, $\RED$ and $x_{d,i}$.
Hence, at this point each $L_{d,i}$ for which the literal $\literal{d}{i}$ is satisfied in $C$ is clean.
Place the two searchers of color $\NPCclause{d}$ in each remaining contaminated $L_{d, i}$ on vertex belonging to $A_{0}$. Because at least one  subtree $L_{d, i}$ is clean for each $d\in\{1,\ldots,m\}$, two searchers of color $\NPCclause{d}$ are sufficient. Then, continue cleaning $A_{0}$ towards the next $v_{h}$. 

\item\label{s:step4} We now describe the moves of the search strategy performed once a vertex $v_j$, $j \in R$, is reached while cleaning $A_0$.

Clean all remaining contaminated edges of $S_{p}$ and $S_{-p}$ rooted in $v_j$. Remove the searchers of colors $\NPCtrue{p}, \NPCfalse{p}, \NPCvariable{p}$ 
when they are no longer necessary.

\item Once $v_{a+1}$ has been reached, clean all remaining contaminated subtrees $L_{d, i}$ (it follows directly from previous steps that searchers of appropriate colors are available) and perform moves as in colorless strategy with the addition of necessary switching of searchers on vertices with multiple colors. Repeat this strategy for each $L_{x}, x \in\{ \NPCtrue{i}, \NPCfalse{i} \st i\in \{1,\ldots, n\}\}$, and $L'_{x'}, x' \in \{\NPCclause{1},\ldots,\NPCclause{m}\}$.

\item For each color $z\in\cQ\setminus\{\RED\}$ set a searcher of color $z$ on the common vertex of the subtree $H_{z}$ and $A_{0}$. Then clean all edges incident to each vertex $v_{i}, i\in L$, which finishes cleaning $A_{0}$. The finishing touches of our strategy are simple.
Once $A_{0}$ is clean, the remaining contaminated parts of subtrees $H_{z}$, containing $A_{i}$, $i\in L$, can be searched in an arbitrary sequence.
\end{enumerate}
\end{proof}

Now we will give a series of lemmas that allow us to prove the other implication, namely that a successful monotone strategy implies a valid solution to $\problemSAT$.
The next lemma says that between the first and last moves when all searchers of color $\RED$ are in an area $A_i$, no move having all searchers of color $\RED$ in a different area $A_j$ is possible.
The proof is due to a counting argument.
\begin{lemma} \label{lemma:sequentialy}
No step $t_{j}$ can occur between any two steps $t_{i}, t'_{i}$:
\[\left[ t_{i}, t'_{i} \right] \cap \left[t_{j}, t'_{j} \right]=\emptyset, \quad i\neq j.\]
\end{lemma}
\begin{proof}
The proof is by contradiction.
First note that if $t_{i}'=t_{j}$ for any $i\neq j$, that would imply that four searchers of color $\RED$ are present in a graph at once: two in $A_i$ and two in $A_j$. Hence, we suppose for a contradiction that there exists $j\neq i$ such that $t_i<t_j<t_i'$ or $t_i<t_j'<t_i'$. Consider a step $t\in \left[t_{i}, t'_{i} \right]$. At least one searcher  of color $\RED$ is in $A_{i}$ because it is partially clean. If $t\in \{t_{j}, t'_{j}\}$, then there are two red  searchers in $A_{j}$, which contradicts color composition imposed by Lemma~\ref{lemma:colors}.
\end{proof}

We say that a subtree $T'$ is \emph{guarded} by a searcher $q$ on a vertex $v$ if removal of $q$ leads to recontamination of an edge in $T'$. Note that $v$ does not have to belong to $T'$, or in other words, $T'$ is any subtree of the entire subgraph that becomes recontaminated once the searcher $q$ is removed. We extend our notation to say that $T'$ is \emph{guarded from} $T''$ if $T''$ is contaminated and removal of $q$ produces a path that is free of searchers and connects a node of $T'$ with a node of $T''$.

Informally, the next lemma states the following.
Prior and after the moves that have all searchers of color $\RED$ on $A_0$, there must be moves having all searchers of color $\RED$ on some $A_i$, $i\neq 0$.
The argument is due to the fact that we do not have sufficiently many searchers, in total, to guard $A_0$ from all other $A_i$'s (or, conversely, all other $A_i$'s from $A_0$).
\begin{lemma}\label{lemma:middle} The step $t_{0}$ cannot be the first one and $t'_{0}$ cannot be the last one in a sequence of steps containing each $t_i$, $i\in I$, i.e., $\min\{t_i\st i\in I\}<t_0\leq t_0'<\max\{t_i\st i\in I\}$.
\end{lemma}
\begin{proof}
Suppose for a contradiction that $t_{0}$ is the first step, i.e., $t_{0}< t_{i}$, for each $i\in I\setminus\{0\}$. By Lemma~\ref{lemma:sequentialy}, $t_{i}> t_{0}'$ for each $i\in I\setminus\{0\}$. 
Thus, each area $A_i$ contains contaminated edges in step $t_0$. In step $t'_{0}$ all edges of the path $P$ are clean.
Hence, $P$ is guarded by at least one searcher from $A_{i}$ for each $i\in I\setminus\{0\}$.
A simple counting argument implies that two searchers of color $\NPCvariable{p}, p\in\{1,\ldots,n\}$, are used --- a contradiction. The second case, when $t_{0}'$ is the last step, can be argued analogously.
\end{proof}
We draw attention to the two ways of searching $A_{0}$ (to which we refer as a folklore).
Since $A_0$ is a caterpillar, one can assume without loss of generality that it is cleaned by $\cS$ by the two searchers of color $\RED$ in the following way.
Either, the first searcher of color $\RED$ is placed, in some move of $\cS$, on $v_1$ and throughout the search strategy it moves along $P$ from $v_1$ to $v_l$ --- we say that such $\cS$ \emph{cleans $P$ from $v_1$ to $v_l$}, or the first searcher starts at $v_l$ and moves along $P$ from $v_l$ to $v_1$ while $\cS$ proceeds --- we say that such $\cS$ \emph{cleans $P$ from $v_l$ to $v_1$}.
In both cases, the second searcher  of color $\RED$ is responsible for cleaning edges incident to $v_i$, $i\in\{1,\ldots,l\}$, when the first searcher is on $v_i$.

We say that an edge search strategy $\cS'$ is a \emph{reversal of} a search strategy $\cS$ that consists of $l$ moves if it is constructed as follows: if the move $i$ of $\cS$ places (respectively removes) a searcher on a node $v$, then the move $(l-i+1)$ of $\cS'$ removes (respectively places) the searcher on $v$, and if the move $i$ of $\cS$ slides a searcher from $u$ to $v$, then the move $(l-i+1)$ of $\cS'$ slides the searcher from $v$ to $u$.
It has been proved in \cite{WormanYang08} that if $\cS$ is a (monotone) edge search strategy, then $\cS'$ indeed is a (monotone) edge search strategy.
We skip a proof (it is analogous to the one in \cite{WormanYang08}) that if $\cS$ is a monotone search $\colS$-strategy, then its reversal is also a monotone search $\colS$-strategy.
This allows us to assume the following for the search strategy $\cS$ for $\TSAT$ we consider in this section:
\begin{enumerate} [label={\normalfont{(*)}},leftmargin=*]
\item\label{eq:reversability} $\cS$ cleans $P$ from $v_l$ to $v_1$.
\end{enumerate}

Let $\first{i}, i \in \{1,\ldots, l\}$, be the index of the first move such that two searchers  of color $\RED$ are in $v_{i}$ (either at the start or end of the move). 
Such moves are well defined because the degree of $v_{i}$ is greater than two.
Note that without loss of generality due to~\ref{eq:reversability}, $\first{l}=t_{0}$. 

Observe that the removal of edges $\{v_a,v_{a+1}\}$ and $\{v_{a+1},v_{a+2}\}$ from $\TSAT$ gives three connected components and let $T_{a+1}$ be the subtree of $\TSAT$ that equals the connected component that contains $v_{a+1}$.
For each subtree $T'$, let $\clean{T'}{t}$ ($\cont{T'}{t}$, respectively) denote the set of clean (contaminated, respectively) edges  in  $T'$ immediately prior to the move $t$.

Intuitively, Lemma~\ref{lemma:va} says that when we reach the vertices $v_a,v_{a+1},v_{a+2}$ while moving along $A_0$, then the tree $T_{a+1}$ is constructed in such a way that while cleaning it there exists a move in which no searcher is used to guard any $A_i$ with $i\in R$ or any clause component.
The configurations of the colors of searchers for the guarding of $A_i$'s and the clause components are those in the family $\cK$ below.
\begin{lemma}\label{lemma:va}
Let
\[\cK=\{\{\RED, \NPCclause{1}\},\ldots,\{\RED, \NPCclause{m}\}, \{\RED, \NPCvariable{1}, \NPCtrue{1}\}, \ldots, \{\RED, \NPCvariable{n}, \NPCtrue{n}\},\{\RED, \NPCvariable{1}, \NPCfalse{1}\},\ldots,\{\RED, \NPCvariable{n}, \NPCfalse{n}\}\}\]
For each $K\in \cK$, between moves $\first{a+2}$ and $\first{a}$, there exists a move $t_K$ which requires all searchers of colors from the set $K$ to be in $T_{a+1}$.
\end{lemma}
\begin{proof}
An intuition explaining the proof is as follows.
Recall that $T_{a+1}$ consists of multiple copies of subtrees  $L_{x}$ and  $L'_{x'}$. Arguments are the same for both $L_{x}$ and  $L'_{x'}$.
We consider which edges of $T_{a+1}$ are clean in the move $\first{a+2}$: $L_{x}$ it is either contaminated or contains a guarding searcher.
Then we consider which edges of $T_{a+1}$ are clean in the move $\first{a}$: in this case $L_{x}$ it is either clean or contains a guarding searcher.
We count how many guarded $L_{x}$'s can exist in the move $\first{a}$.
A counting argument reveals that at least three $L_{x}$'s are clean in the move $\first{a}$. From all of the above, these 3 subtrees $L_x$ must have been cleaned after $\first{a+2}$. Among the moves of cleaning 3 subtrees $L_x$, a move $t_K$ exists.

Consider what can be deduced about $\clean{T_{a+1}}{\first{a+2}}$ and $\cont{T_{a+1}}{\first{a+2}}$ from~\ref{eq:reversability} and the  definition of $\first{a+2}$.
Recall that by~\ref{eq:reversability}, the strategy we consider cleans $P$ from $v_l$ to $v_1$.
Hence, the edge $\{v_a,v_{a+1}\}$
is contaminated before move $\first{a+2}$. Furthermore, the vertex $v_{a+1}$ cannot be occupied by a searcher during move $\first{a+2}$, because both searchers  of color $\RED$ are on the vertex $v_{a+2}$, by the definition of $\first{a+2}$. Therefore, no subtree $L_{x}, x \in X=\{\NPCtrue{i}, \NPCfalse{i}\st i\in\{1,\ldots,n\}\}$, or $L'_{x'}, x'\in X'=\{\NPCclause{1}, \ldots, \NPCclause{m}\} $, can be fully clean and unguarded in the move $\first{a+2}$, because these subtrees are incident to $v_{a+1}$. For each subtree $L_x$, there are two possibilities: either $\clean{L_{x}}{\first{a+2}}=\emptyset$ or $\clean{L_{x}}{\first{a+2}}\neq\emptyset$ in which case $L_{x}$ contains at least one guarded  vertex. The same holds for any $L'_{x'}$. 

Next consider what is known about $\clean{T_{a+1}}{\first{a}}$ and $\cont{T_{a+1}}{\first{a}}$. Because $v_{a+1}$  is not  guarded in the move $\first{a}$ as both searchers of color $\RED$ are in vertex $v_{a}$, the vertex $v_{a+1}$ is not incident to contaminated edges at this point. Otherwise all edges of $P$ connecting $v_{a+1}$ and $v_{l}$ would be contaminated which contradicts~\ref{eq:reversability}.  The spread of contamination at move $\first{a}$  from each $L_{x}$ and $L'_{x'}$ through $v_{a+1}$ can be prevented only in the following way: $\clean{P}{\first{a}}$  is guarded from the contaminated edges in subtrees $L_{x}$ and $L'_{x'}$ and their copies, and all remaining subtrees in $T_{a+1}$  are clean. Again there are two possibilities: either $\clean{L_{x}}{\first{a}}=E(L_{x})$ or if $\clean{L_{x}}{\first{a}}\neq E(L_{x})$, then $L_{x}$ contains at least one guarded  vertex. The same holds for any $L'_{x'}$.

By the above paragraphs, each subtree $L_{x}$ and $L'_{x'}$ at some point between $\first{a+2}$ and $\first{a}$ is either being guarded from or is fully searched.
Suppose for a contradiction, that three out of five copies of a subtree  $L_{x}$ or $L'_{x'}$ contain a guarding searcher in the move $\first{a}$.
A simple counting argument suffices to show that they have to be guarded on a common vertex: for $L_{x}$, there are two  searchers of color $\RED$, and they are placed on the vertex $v_a$, and one in each of the colors $x$ and $\NPCvariable{i}$ where  $i$ is selected so that $x\in\{\NPCtrue{i}, \NPCfalse{i}\}$. 
Similarly for any $L'_{x'}$ there are only two searchers of color $x'$ which can be placed in it. The only common vertex is $v_{a+1}$ --- contradiction with the definition of $\first{a}$.

Because we eliminated the option of guarding three subtrees $L_{x}$ and $L'_{x'}$ at the move $\first{a}$,  the only remaining possibility is that at least three of the subtrees $L_{x}$ and $L'_{x'}$ are clean before the move $\first{a}$. Consider a move after which the first subtree has been cleaned. This subtree is guarded on $v_{a+1}$ until all edges connected to $v_{a+1}$ are clean, so in subsequent moves at least one of the copies of each $L_{x}$ and $L'_{x'}$ is cleaned while   $v_{a+1}$ is guarded by a searcher of color $\RED$.
By construction, for each of the following sets:  $B\in\cB=\{\{\RED, \NPCvariable{1}, \NPCtrue{1}\}, \ldots, \{\RED, \NPCvariable{n}, \NPCtrue{n}\}\}$ and $F\in\cF=\{\{\RED, \NPCvariable{1}, \NPCfalse{1}\}, \ldots, \{\RED, \NPCvariable{n}, \NPCfalse{n}\}\}$ there exist a  subtree $L_{x}$ requiring searchers in these colors. Note that $\sn{L_{x}}=3$, so all of those searchers will be required simultaneously in at least one move, whose number is denoted by $t_{B}$ or $t_{F}$ respectively, when $L_{x}$ is being searched. Similarly $L'_{x'}$ will require all searchers of colors $G\in\cG=\{\{\RED, \NPCclause{1}\},\ldots,\{\RED, \NPCclause{m}\}\}$ to be present in $T'$ in a single move, whose number is denoted by $t_G$. $\cB\cup \cF\cup \cG=\cK$, so $t_K$ exists for each $K\in \cK$.
\end{proof}

\begin{lemma}\label{lemma:va2}
Let
\[\cK''=\{\{\RED, \NPCclause{1}\},\ldots,\{\RED, \NPCclause{m}\}\}\]
For each $K''\in \cK''$, between moves $\first{l}$ and $\first{l-2}$, there exists a move $t_K''$ which requires all searchers of colors from the set $K''$ to be in one of the subtrees $L''_{\NPCclause{d}}, d\in\{1,\ldots, m\}$.\qed
\end{lemma}
We skip the proof because it is analogous to the one of Lemma \ref{lemma:va}. 

\medskip
Let $\cT$ be the set of subtrees $H_{z}, z \in \cQ\setminus\{\RED\}$, and $S_{p}, S_{-p}, p\in\{1,\ldots, n\}$. For $G\in \cT$ let $\tau(G)$ be the index $i$ such that $G$ contains the vertex $v_{i}$.

We say in Lemma~\ref{lemma:sequence}, informally, that we need to entirely clean all areas $A_i$ with $i\in R$ prior to the part of the search strategy that uses all searchers of color $\RED$ on $A_0$.
The latter part is the one that cleans $A_0$ entirely.
\begin{lemma} \label{lemma:sequence}
The step $t_{0}$ is placed in the search sequence in the following way:
\[t_{j}\leq t'_{j}< t_{0}\leq t'_{0} \]
for each  $j\in R$.
\end{lemma}
\begin{proof}
We first summarize the intuitions used in the proof. 
We start by using Lemma~\ref{lemma:sequentialy} and~\ref{lemma:middle}, which give us that $A_0$ is not the first nor the last area $A_i$ cleaned.
We define two sets of numbers ($U^-$ and $U^+$ below) corresponding to indices of $A_i$'s whose cleaning happens before and after cleaning $A_0$. During the moves $t_0,\ldots,t_0'$ (i.e., those that clean $A_0$) the clean subgraph of $A_0$ has to be guarded from the contaminated $A_i$'s, and the cleaned $A_i$'s have to be guarded from the contaminated part of $A_0$ Intuitively, once cleaning of $A_0$ extends past a vertex $v_i$ to which a subtree containing $A_i$ is attached, this $A_i$ needs to be guarded if it's `status' (being clean or contaminated) is different than that of the area $A_0$.
Consider a move, which we denote by $l_{a+1}$ below in the proof, in which the vertex $v_{a+1}$ divides the clean and contaminated parts of $P$.
We analyze which $A_j$'s, $j\in R$, could have been left contaminated and which ones are guarded in the move $l_{a+1}$.
We obtain that there exists a move $t_K, K\in\cK$, described in the Lemma~\ref{lemma:va}, which can be identified with $l_{a+1}$.
There is not enough searchers to perform $t_K$ while $A_j$ is guarded --- a contradiction.

Now we start the formal proof.
Define
\[U^-=\{i\in I\st t_i\leq t_i'<t_0\}\]
and
\[U^+=\{i\in I\st t_0'<t_i\leq t_i'\}.\]
By Lemmas \ref{lemma:sequentialy} and~\ref{lemma:middle}, $U^-\neq\emptyset$, $U^+\neq\emptyset$ and $U^-\cup U^+=I\setminus\{0\} $.
Note that 
$U^-\cap U^+= \emptyset$ because the strategy is monotone.
Given this notation we restate the lemma as $U^-$ contains all indices of  steps $t_{j}, j\in R$, i.e., $R\subseteq U^-$. 

 Let $u^{-}\in U^-$ and $u^{+}\in U^+$ be selected arbitrarily. Let $g_{\left | i \right |}, i\in I$ be the index of the first move in $\left [ t_{0}, t_{0}'\right ]$ such that $v_{\left | i \right |}$ is incident to a clean edge. There exists $G\in\cT$ such that $\tau(G)=  \left | u^{+} \right | $ and $G$ has a contaminated edge between moves of numbers $g_{\left | u^{+} \right |}$ and $t_{0}'$ because $t_0'<t_{u^{+}}$.
Thus, $\clean{A_0}{t}, t\in \left[g_{\left | u^{-} \right |}, t_{0}'\right] $, is guarded from  contaminated edges of $G$. 

Let $g_{\left | i \right |}', i\in I$, be the index of the last move in $\left [ t_{0}, t_{0}'\right ]$ such that $v_{\left | i \right |}$ is incident to a contaminated edge on $P$.
There exists $G'\in\cT$ such that $\tau(G')=  \left | u^{-} \right | $ and $G'$ has a clean edge between steps $t_{0}$  and $g_{\left |u^{-} \right |}$ because $t_{u^{-}}<t_0$.
Thus, $\clean{G'}{t}, t\in \left[t_{0}, g'_{\left | u^{-} \right |}\right]$, is guarded from  contaminated edges of $P$.

Let $l_{a+1}$ be an arbitrary move number such that the edge $\{v_{a}, v_{a+1}\}$ is contaminated and the edge $\{v_{a+1}, v_{a+2}\}$ is clean.
By $(*)$ such a move exists.

Suppose for a contradiction that $R\nsubseteq U^-$, which is equivalent to $R\cap U^+\neq\emptyset$.
 During the move $l_{a+1}$ the following subtrees are guarded: $$\clean{A_0}{l_{a+1}} \quad from \quad  A_{j}, j\in R\cap U^+
, \quad because \quad \forall_{ j\in R\cap U^+ } \quad g_{\left | j \right |}<l_{a+1}$$ and 
$$\clean{A_{i}}{l_{a+1}} \quad from \quad \cont{A_0}{l_{a+1}}, i\in L\cap U^-
\quad because  \quad \forall_{ i\in U^-\cap L } \quad l_{a+1}<g'_{\left | i \right |}$$

By Lemma \ref{lemma:va}, there exists a move $t_K, K\in\cK$,  which requires  all searchers of colors belonging to $K$ to be present in $T_{a+1}$. 
Every edge incident to $v_{a}$ cannot be clean before the last move which places two searchers of color $\RED$ in  $v_{a}$ has occurred, therefore $\first{a}\leq g_{a}'$. Two searchers of color $\RED$ cannot be placed in $v_{a+2}$ before at least one edge incident to it is clean, therefore $g_{a+2}\leq \first{a+2}$. This gives us $g_{a+2}\leq \first{a+2} <t_K< \first{a}\leq g'_{a}$ , and with the fact that in the moves $t_K$ and $l_{a+1}$ the vertex $v_{a+1}$ is occupied, allows us to conclude that for each $t_K$ there exists $l_{a+1}$ such that $t_K=l_{a+1}$.

Let $G\in\cT$ be such that $\tau(G)>a+1$. Recall that if $j\in R$ and $\tau(G)=\left | j \right |$, then $G$ is isomorphic to some $S_{p}$ or $S_{-p}$. Due to construction of $S_{p}$ and $S_{-p}$, searcher used to guard $\clean{A_{0}}{l_{a+1}}$ from $\cont{G}{l_{a+1}}$ is in one of the following colors: $\RED$, $\NPCvariable{p}$, $\NPCtrue{p}$ if $j>0$ or $\RED$, $\NPCvariable{p}$, $\NPCfalse{p}$ if $j<0$. 
Let $D_G$ denote a set of colors in $G$ for each $G\in\cT$.
Note that for each $D_G$ there exists $K\in \cK$ such that $D_G\subseteq K$,  and therefore there exists a move $t_{D_{G}}$  such that all searchers in colors $D_{G}$ are in $T'$.
Consider a move $t_{D_{G}}=l_{a+1}$, which requires a searcher of one of the colors in $D_{G}$ to be present in $G$, such that it contains an area  $A_w, w\in R\cap U^+$. Because $T'$ and $G\in\cT$ contain no common vertices, such a move cannot exist --- a contradiction with the Lemma \ref{lemma:va}. 
\end{proof}

The statement of the following lemma is more specific with respect to the previous one: in Lemma~\ref{lemma:guard} we examine those moves in which each area $A_j$, $j\in R$, is clean and guarded from $A_0$. Additionally we are concerned with the colors of searchers guarding these areas.
More precisely, between the time when cleaning of $A_0$ starts and reaches $v_b$, all  subgraphs $A_i, i\in R$, contain clean edges which are guarded from the contaminated edges of $A_0$. Only searchers of colors $\RED$ and either $\NPCtrue{p}$ or $\NPCfalse{p}$ can be used for the guarding.

\begin{lemma}\label{lemma:guard}
In an arbitrary move $t\in[t_{0},\first{b} ]$, each subgraph $\clean{A_j}{t}$ for each $j\in R$ is guarded from $\cont{A_0}{t}$  by searchers on at least one of these two vertices: a vertex with colors $\{\NPCtrue{p}, \RED\}$ in $S_{p}$ or $\{\NPCfalse{p},\RED\}$ in  $S_{-p}$
\end{lemma}
\begin{proof}
Informally, we analyze the state of the strategy when $A_0$ is being cleaned but $v_b$ has not been reached, i.e., in a move $t$ from the lemma.
In the moves $t_{0}$ and $\first{b}$ each of $S_p$ and $S_{-p}$ has some clean edges that need to be guarded.
In the proof, we consider several cases as to which vertices can be guarded to protect those clean edges in moves  $t_{0}$ and $\first{b}$.
Then, we observe that, due to monotonicity, the set of clean edges when the two searchers of color $\RED$ are on $v_b$ could not be smaller than in any move prior to it. Thus, 
if both a vertex with colors $\{\NPCtrue{p}, \RED\}$ and a vertex with colors $\{\NPCfalse{p},\RED\}$ need not be guarded between these moves, then a recontamination occurs in the move $\first{b}$ which leads to a contradiction.

From Lemma~\ref{lemma:sequence} we have $t_{j}\leq t'_{j}<t_{0}\leq t'_{0} $ for each $j \in R$. 
All stars of color $\RED$ in subtrees $S_{p}$ and $S_{-p}$ attached to vertices $v_{\left |j\right|}, j \in R$,  contain clean edges before the move $t_{0}$, so at the move $t_{0}$ each such star is guarded from contaminated edges in $A_{0}$.

Consider the moves number $t_{0}$ and $\first{b}$. In these moves both searchers of color $\RED$  are in $A_{0}$, and $v_{\left |j\right|}$ are not attached to any clean edges, so guarding searchers are still necessary at move $\first{b}$. Recall that by the definition of $R$, $b\geq \left| j \right |$.
Because searchers of color $\RED$ are in $A_{0}$, the vertices with the following colors are occupied by searchers: $\NPCfalse{p}$ or $\NPCvariable{p}$ in $S_{-p}$ and $\NPCvariable{p}$ or $ \NPCtrue{p}$ in $S_{p}, p\in\{1,\ldots, n\}$. 
Because only one searcher of color $\NPCvariable{p}$ is available, at least one other searcher is placed on vertex with either $\NPCfalse{p}$ or $\NPCtrue{p}$ color, denoted by $c$.
In each $S_{p}$ and $S_{-p}$ there are only two such vertices separated by an edge of color $c$.
It remains to be proven that at least one of them has to be guarded in an arbitrary move $t\in [t_{0},\first{b} ]$.

Assume for a contradiction that both vertices with color $\NPCfalse{p}$ and $\NPCtrue{p}$ are no longer guarded in a move $t\in[t_0, \first{b}]$, thus
all edges incident to vertices with color $c$ are clean. By monotonicity, all these edges are clean in the move $\first{b}$. Edges of color $\RED$ incident to the vertex $v_i$ such that $\tau(G)=i$ are contaminated.
Since the two searchers of color $\RED$ are in $v_b$ in the move $\first{b}$, they are not in $G$. Thus, the searcher of color $c$ is the only one that can be used for guarding $\clean{G}{\first{b}}$ from $\cont{A_0}{\first{b}}$. Since all edges incident to the vertex  occupied by this searcher are clean, some recontamination occurs.
\end{proof}

\begin{lemma}\label{lemma:3SAT}
Let $x_{1},\ldots, x_{n}$ and a Boolean formula $C=C_{1}\land C_{2}\ldots\land C_{m}$ be an input to $\problemSAT$.
If there exists a search strategy using $2+3n+2m$ searchers for $\TSAT$, then the answer to $\problemSAT$ is $\YES$.
\end{lemma}
\begin{proof}
The proof revolves around the configuration of searchers in the move  $\first{b}$.
We start by recalling the construction of subtrees based on clauses that must have been cleaned up to this point and the colors of searchers required to clean them. Then, we will use Lemma~\ref{lemma:guard} to address the availability of these colors.
We define a Boolean assignment as follows: $x_{p}$ is true if and only if a searcher of color $\NPCtrue{p}$  does \emph{not} guard the area $A_{a+1+p}$ in the move $\first{b}$, otherwise $x_p$ is false. Let $X\subset \{\NPCtrue{i}, \NPCfalse{i}\st i\in\{1,\ldots,n\}\}$ denote the colors of those searchers.
By Lemma \ref{lemma:guard}, a valid assignment will occur during execution of arbitrary successful search strategy using $2+3n+2m$ searchers.
We argue that some literal in each clause $C_{d}, d\in \{1,\ldots, m\}$, is true under the above assignment.

By construction of $\TSAT$, $v_{h}, h \in \{b+1,b+2,\ldots, l-1\}$, is the root of a subtree $L_{d, i}, i \in\{1,2,3\}$. $L_{d, i}$ contains an edge of color $\NPCtrue{p}$ if and only if the clause $C_{d}$ contains a variable $x_{p}$, and it contains an edge of color $\NPCfalse{p}$ if and only  if $C_{d}$ contains a variable's negation, $\overline{x_{p}}$. Let us denote this color as $x_{d, i}$.

Consider the step $t_{0}$.  It is impossible for any $L_{d, i}$ to be completely clean because all edges incident to $v_{h}, h \in \{b+1,b+2,\ldots, l-1\}$, are contaminated (there are not enough searchers of color $\RED$). 
Consider the move $\first{b}>t_{0}$. Each $L_{d, i}$ is completely clean or $\clean{A_0}{\first{b}}$ is guarded by searchers of colors $x_{d, i}$ and $\NPCclause{d}$ because $v_{h}$ is connected to clean edges and unguarded (there are not enough searchers of color $\RED$).
There are only two  searchers of color $\NPCclause{d}$, so for each $d$ one of the subtrees $L_{d, i}$ has a searcher of other color or is clean. Because $\sn{L_{d, i}}=3$, three out of five searchers in the following colors can be used: $\RED, x_{d, i},\NPCclause{d}$. The only searcher which is present in $L_{d, i}$ at move $\first{b}$ has color $x_{d, i}$.
Without loss of generality we assume that the strategy cleans a subtree if possible before guarding other subtrees rooted in the same vertex. 
Consider a move $m_{d,i}$ such that $\sn{L_{d, i}}$ searchers are used in $L_{d, i}$.
Due to the way the subtree is colored, a searcher in each color $\RED, x_{d, i},\NPCclause{d}$ has to be used in order to clean it. After $t_{0}$ a  searcher of color $\RED$ guards clean part of $A_{0}$, so only one searcher out of those five, namely the one of color $\NPCclause{d}$, can be present outside of $L_{d, i}$. $L_{d, i}$ can be fully cleaned only if 
searcher of color $x_{d, i}$ is available at this point during $[t_{0}, \first{b}]$, or all edges of color  $x_{d, i}$ were clean in the move $t_0$ . Due to its color this searcher can not be used to replace  
any searcher of color $x_{d, i}$ outside of $L_{d, i}$.

Let us address what follows if an edge of color  $x_{d, i}$ was clean in the move $t_0$. By construction, it can be guarded from contaminated edges of $P$ by a searcher of one of the following colors: $x_{d, i}$, $\NPCclause{d}$, $\RED$ in the moves of numbers from the interval $[\first{l}, \first{l-2}]$.
By Lemma~\ref{lemma:va2}, there exists a move of number in this interval such that all searchers of colors  $\NPCclause{d}$ and $\RED$ are not in $L_{d, i}$. Thus, in order to avoid recontamination, clean edges of color $x_{d, i}$ can be guarded only by a searcher of the same color.

Suppose for contradiction that no literal in a clause $C_i, i\in \{1,\ldots, n\}$ is true and
a search strategy for $\TSAT$ exists. %\RO{Recall that the path composed of the vertices $v_{b+1},\ldots,v_{l-1}$ is being cleaned in moves between $t_0$ and $\first{b}$.} 
By Lemma~\ref{lemma:guard}, one of the searchers of color $x\in\{ \NPCtrue{p}, \NPCfalse{p}\}$,  or a searcher of color $\RED$ is placed outside of $L_{d, i}$ during $[t_{0}, \first{b}]$. By the definition of a Boolean assignment $x\in X$. Additionally Lemma~\ref{lemma:guard} guarantees that no searcher of color $x$ is used during cleaning any $L_{d, i}$.
If $x=x_{d, i}$ then $m_{d, i}$ can not be performed and $L_{d, i}$ is guarded in the move $\first{b}$. 
Because there are only two searchers of color $\NPCclause{d}$, in order for a strategy to exist at least one subtree $L_{d, 1}$, $L_{d, 2}$, $L_{d, 3}$ for each $d$ is cleaned before $\first{b}$, or all three have to be guarded, and for that to happen they have to contain edges in at least one color $x_{d, i}\neq x$  corresponding to a true $\literal{d}{i}$ in the clause $C_{d}$.
\end{proof}

\section{NP-hardness of non-monotone searching of trees} \label{sec:nmhard}
This section is devoted to proving the problem remains NP-hard  when non-monotone search strategies are allowed:
\begin{theorem}\label{nph_t_theorem}
The problem $\problemHGS$ is NP-hard in the class of trees.
\end{theorem}
For the proof, we adapt the tree $\TSAT$ described in the previous section. The modified tree is denoted by $\TSATP$ and it is obtained by performing the following operations on the tree $\TSAT$. In order to preserve the familiar notation we denote each component of $\TSATP$ analogous to its counterpart in $\TSAT$ with an additional sign $\sim$ above its designation. 

We add $4n$ vertices to the path $P$ in the following fashion. Replace the edges $\{v_{a+1}, v_{a+2}\}$ and $\{v_{a+1}, v_{a}\}$ with paths of color $\RED$ of length $2n$ each, denoted by $\tilde{P}_R$ and $\tilde{P}_L$ respectively. Enumerate the vertices of $\tilde{P}$ in  $\TSATP$ as $\tilde{v}_i$ in such a way that $\tilde{v}_1=v_1$, $\tilde{v}_{\tilde{a}+1}=v_{a+1+2n}$, $\tilde{v}_{\tilde{b}}=v_{b+1+4n}$, $\tilde{v}_{\tilde{l}}=v_{l+4n}$. Note that this enumeration preserves the informal division of vertices into sets on the left and right of $\tilde{v}_{a+1}$, and $\tilde{R}=\{ \tilde{a}+2,\ldots,\tilde{b}-1 \} \cup \{-(\tilde{b}-1),\ldots,-(\tilde{a}+2) \}$ is defined accordingly.

We use $2n$ additional colors $O=\{\NPCvalve{1}{1},\ldots,\NPCvalve{n}{1},\NPCvalve{1}{2},\ldots, \NPCvalve{n}{2}\}$. For each $o\in O$ create a tree $H_{o}$ following the construction defined in the previous section and attach one to a unique vertex of the path $\tilde{P}_L$. We do the same for the path $\tilde{P}_R$ so that $4n$ subtrees are created in total. Let $\tilde{H}_o( \tilde{R})$ denote a subtree containing an edge of color $o\in O$ attached to the vertex $\tilde{v}_i,  i\in \tilde{R}$.  

Next we modify the construction of each subtree $L_x, x\in\{\NPCtrue{i},\NPCfalse{i}\mid i\in\{1,\ldots, n\}\}$  rooted in $v_{a+1}$ in the following way. Remove 2 leaves of color $\RED$ and attach 3 children by the edges of color $\NPCvalve{i}{1}$ to each leaf. Then attach 3 children by the edges of color $\NPCvalve{i}{2}$ to each of the new leaves.
Finally attach 2 children by the edges of color $\RED$ to each of the lastly added leaves. Denote the modified $L_x$ as $\tilde{L}_x$. Note that $\sn{\tilde{L}_x}=5$ and $\tilde{L}_x$ requires searchers of colors $x,\RED, \NPCvariable{i}, \NPCvalve{i}{1}, \NPCvalve{i}{2}$  to be simultaneously present in some move in $\tilde{L}_x$ in order to search it. In $\TSATP$, eleven copies of $\tilde{L}_x$ are rooted in $\tilde{v}_{a+1}$ in place of five copies of $L_x$ rooted in $v_{a+1}$ in the original $\TSAT$. Whenever an arbitrary copy can be chosen the notation of $\tilde{L}_x$ is used, when the argument requires copies to be distinct they are denoted by $\tilde{L}_{x, i}, i\in\{1,\ldots,11\}$.

Define a star $O_p, p\in\{1,\ldots, n\}$ with 3 leaves incident to edges of colors:  $\NPCvalve{p}{1}$, $\NPCvalve{p}{1}$ and  $\NPCvalve{p}{2}$.
We modify each subtree $\tilde{S}_p$ in $\TSATP$ corresponding to $S_p$ constructed according to the definition in the previous section. Recall that each $S_{p}$ and $\tilde{S}_{p}$ (respectively $S_{-p}$ and $\tilde{S}_{-p}$) contains an edge of color $\NPCtrue{p}$ ($\NPCfalse{p}$ respectively).
Define a $plugin(v, u)$ operation for vertices $u$ and $v$ of a tree as replacement of maximal subtree such that $u$ and $v$ are its leaves with a copy of $O_p$ in such a way that $u$ is identified with a leaf of color $\NPCvalve{p}{1}$ and $v$ is identified with a leaf of color $\NPCvalve{p}{2}$. 
For each $\tilde{T}\in\{\tilde{S}_p, \tilde{S}_{-p}\st p\in\{1,\ldots,n\}\}$ denote the endpoint which belongs to $\tilde{A}_0$ of the edge of color $\NPCtrue{p}$ or $\NPCfalse{p}$ in $\tilde{T}$ as $u_1$, and the other endpoint of this edge as $u_2$. Denote the endpoint which belongs to $\tilde{A}_{j}, j\in \tilde{R}$, of edge of color $\NPCvariable{p}$ as $u_3$, and the other endpoint of this edge as $u_4$. Perform $plugin(u_1, u_1)$, $plugin(u_2, u_3)$ and $plugin(u_4, u_4)$.

Informally speaking, the described modification prevents using recontamination to switch the searchers used as a basis for Boolean assignment without undoing all the progress made while cleaning subtrees corresponding to the clauses.

The following lemma follows directly from the lower bound $\lb{\TSATP}$ and the proof is analogous to Lemma~\ref{lemma:colors}
\begin{lemma}[Color assignment]\label{lemma:colorsnm}
A $\colS$-strategy using $\anum=5n+2m+2$ searchers has to color them in the following fashion: one searcher for each color in $\{\NPCtrue{p},\NPCfalse{p},\NPCvariable{p}, \NPCvalve{p}{1},\NPCvalve{p}{2}\st p\in\{1,\ldots,n\}\}$ and two searchers for each color in $\{\RED\}\cup\{\NPCclause{1}, \ldots, \NPCclause{m}\}$. 
\qed
\end{lemma} 

\begin{lemma}\label{lemma:strategynm}
Let $x_{1},\ldots, x_{n}$ and a Boolean formula $C=C_{1}\land C_{2}\ldots\land C_{m}$ be an input to the $\problemSAT$.
If the answer to $\problemSAT$ is $\YES$, then there exists a search strategy using $5n+2m+2$ searchers for $\TSATP$.
\end{lemma}
\begin{proof}
We propose a modification to the monotone strategy described in Lemma \ref{lemma:strategy}. Note that the modified strategy is still monotone (we aim to show that recontamination does not help to search $\TSATP$). 
In the instruction~\ref{s:step2} we clean $\tilde{A}_{-(\tilde{a}+2n+1+p)}$ instead of $A_{-(a+1+p)}$ and stars $O_{p}$ instead of singular  edges of color $\RED$ replaced by these stars during construction of $\TSATP$.  No searcher of color either $\NPCvalve{p}{1}$ or  $\NPCvalve{p}{2}$ has already been placed on $\TSATP$ so it is  always possible. We introduce an additional instruction 2' executed after the instruction number~\ref{s:step2}. 
\begin{description}
\item [2'.] For each $\tilde{H}_{o}(\tilde{R})$ place a searcher of color $o$ on the vertex
of color $o$  and belongs to an $\tilde{A}_{i}$ in $\tilde{H}_{o}(\tilde{R})$. Then, clean each $A_{i}$, where $ \tilde{v}_i\in \tilde{P}_R$, and then the edge of color $o$. The searcher of color $o$ stays in $\tilde{H}_{o}(\tilde{R})$.
\end{description}
In instruction \ref{s:step4} a searcher of color $o$ is removed from $\tilde{H}_{o}(\tilde{R})$ when  the entire $\tilde{H}_{o}(\tilde{R})$ becomes clean during cleaning of $\tilde{A}_0$ in order to ensure that $\tilde{L}_x$ can be searched. 
\end{proof}

\subsection{Preliminaries on non-monotone strategies for $\TSATP$}
%\RO{
%(zbyt nieformalnie)Because the relatively simple condition of showing that recontamination is unavoidable is insufficient while examining non-monotone strategies we introduce a notation which aims to describe progress  
%}

Let $G'$ be a subgraph of $G$. We define a \textit{successful attempt} $\atts{G'}{i}{S}=[t,t']$ as a maximal interval of numbers  of moves such that for each  $j\in [t,t']$, $\cont{G'}{j}\neq G'$, $\cont{G'}{j}\neq \emptyset$  and $\clean{G'}{t'}=G'$,  and $i$ is the ordeal number of this attempt among other successful attempts on $G'$. Analogously define an \textit{unsuccessful attempt} $\atts{G'}{i}{U}=[t,t']$  as a maximal interval of numbers  of moves such that for each  $j\in [t,t']$, $\cont{G'}{j}\neq G'$, $\clean{G'}{j}\neq G'$ and $i$ is the ordeal number of this attempt among other unsuccessful attempts on $G'$. 

By definition, at least one edge of $G'$ is clean during an attempt on $G'$. We remove the prefix $U$ or $S$ whenever the success of the attempt is not important at the point of speaking.

Note that any search strategy which cleans a graph $G$ contains the $\atts{G'}{1}{S}$  for any subgraph $G'$, thus in order to show that cleaning a graph is impossible it suffices to prove that there exists a subgraph for which there can be no successful attempt. By strengthening the previous statement we obtain the following:
\begin{observation}
If $G'$ is a subgraph of $G$, then for each $\atts{G}{i}{S}$ there exists $\atts{G'}{j}{S}$ such that $\atts{G'}{j}{S}\subseteq \atts{G}{i}{S}$.
\end{observation}

We skip the proof of the following lemma as it follows from the folklore of the way to clean a caterpillar graph. 

\begin{lemma} \label{lemma:fail}
If in any move of $\att{\tilde{P}}{i}=[t,t']$ both edges $\{\tilde{v}_1, \tilde{v}_2\}$ and $\{\tilde{v}_{\tilde{l}}, \tilde{v}_{\tilde{l}-1}\}$ are contaminated, then the attempt is unsuccessful.\qed
\end{lemma}
% \begin{proof}
% During the $\att{\tilde{P}}{i}$ at least one edge of $\tilde{P}$ is clean, let us denote it as $e$. In arbitrary move number $k: t<k<t'$ such that $\cont{\induced{\TSATP}{\{\tilde{v}_0, \tilde{v}_1, \tilde{v}_{\tilde{l}}, \tilde{v}_{\tilde{l}-1}\}}}{k}$ $e$ is guarded and separates two contaminated subpaths of $\tilde{P}$. It can be guarded only by searchers of color $\RED$  by construction of $\tilde{A}_0$, and guarding $e$ requires all two searchers of color $\RED$. Because of lack of further searchers of color $\RED$ the attempt can not be successful without a move which causes recontamination. Removing a  searcher of color $\RED$ from $\tilde{P}$ causes $\tilde{P}$ to become fully contaminated and makes the attempt unsuccessful. Sliding a searcher along clean edges of $\tilde{P}$ and recontaminating them doesn't change the situation until the two searchers of color $\RED$ are on the same vertex, at which point all edges of $\tilde{P}$ are contaminated.   
% \end{proof}

As a consequence we can be sure that a $\atts{\tilde{P}}{i}{S}$ starts cleaning at either $\tilde{v}_1$ or $\tilde{v}_{\tilde{l}}$ and ends at $\tilde{v}_{\tilde{l}}$ or $\tilde{v}_1$ respectively. If it starts at $\tilde{v}_l$ and if the edge $\{\tilde{v}_{j+1}, \tilde{v}_{j}\}$ is clean, then the next clean edge of $\tilde{P}$ can be only $\{\tilde{v}_{j}, \tilde{v}_{j-1}\}$ or the next contaminated set of edges of $\tilde{P}$ has to include all edges $\{\tilde{v}_{j+1}, \tilde{v}_{j}\},\ldots, \{\tilde{v}_{j+1+x}, \tilde{v}_{j+x}\}$ for some $j+x<l, j-1>0$  and no other edges of $\tilde{P}$. We say that in such attempt a strategy \emph{cleans $\tilde{P}$ from $\tilde{v}_{\tilde{l}}$ to $\tilde{v}_1$}.
Thanks to the result concerning reversal of strategies established in \cite{WormanYang08} a symmetrical case does not need to be considered. Thus, we establish an assumption about non-monotone strategies analogous to ~\ref{eq:reversability}:
\begin{enumerate} [label={\normalfont{(**)}},leftmargin=*]
\item\label{eq:reversability-nm} $\cS$ cleans $\tilde{P}$ in $\F$ from $\tilde{v}_{\tilde{l}}$ to $\tilde{v}_1$.
\end{enumerate} 

Denote the number of a move when two searchers of color $\RED$ arrive on the vertex $\tilde{v}_i$ of the path $\tilde{P}$ in $\F$ as $\f{i}{j}$ where $j$ is the ordinal number of the move $\f{i}{j}$ among other moves $\f{i}{j}$ of index $i$. Specifically $\f{i}{j+1}$ is the number of the first such move after the move  $\f{i}{j}$.
Informally speaking, the $j$ in the expression $\f{i}{j}$ indicates how many times the vertex $\tilde{v}_i$ was reached in the first successful attempt to clean $\tilde{P}$.
Let $\tilde{P}_{i}, i\in \{1, \ldots, \tilde{l}\}$, denote $\induced{\TSATP}{\{\tilde{v}_{\tilde{l}},\ldots, \tilde{v}_{\tilde{i}}\}}$. Let $i(j)=k+j$ denote the ordeal number of the $\atts{G}{i(j)}{S}$ such that $i(0)$ is the ordeal number of the first $\atts{G}{k}{S}$ such that $\atts{G}{k}{S}\subseteq\F$.
Whenever we speak of $\atts{G}{i(0)}{S}$, we are concerned with the first attempt to clean $G$ within the first attempt which successfully cleaned $\tilde{P}$.

\subsection{Some technical lemmas}

In the next lemma we show that before the vertex $\tilde{v}_{\tilde{a}}$ is reached in the first successful attempt to clean $\tilde{P}$, for each of the listed sets of colors there exists a move which requires searchers of those color to be present in  $\tilde{T}_{\tilde{a}+1}$.
Analogously to Lemma~\ref{lemma:va}, these sets correspond to sets of colors of subtrees attached to the vertices of $\tilde{P}_{\tilde{a}}$.
\begin {lemma}\label{lemma:mc}
Let 
\[\cK = \{ \{\RED, \NPCclause{i}\}\st i\in\{1,\ldots,m\}\} \cup \{ \{\RED, \NPCvariable{n},\NPCvalve{n}{1},\NPCvalve{n}{2}, x\} \st x\in\{\NPCtrue{1},\ldots,\NPCtrue{n},\NPCfalse{1},\ldots,\NPCfalse{n}\} \}. \]
% \begin{gather}
% \cK=\{\{\RED, \NPCclause{1}\},\ldots,\{\RED, \NPCclause{m}\}, \\ 
% \{\RED, \NPCvariable{1},\NPCvalve{1}{1},\NPCvalve{1}{2}, \NPCtrue{1}\}, \ldots, \{\RED, \NPCvariable{n},\NPCvalve{n}{1},\NPCvalve{n}{2}, \NPCtrue{n}\},\\ \{\RED, \NPCvariable{1},\NPCvalve{1}{1},\NPCvalve{1}{2}, \NPCfalse{1}\},\ldots,\{\RED, \NPCvariable{n},\NPCvalve{n}{1},\NPCvalve{n}{2}, \NPCfalse{n}\}\}
% \end{gather}
For each $K\in \cK$, in $\atts{\tilde{P}_{\tilde{a}}}{i(0)}{S}$, there exists a move $t_K\leq\f{\tilde{a}}{1}$ which requires all searchers of colors from the set $K$ to be in $\tilde{T}_{\tilde{a}+1}$. 
\end {lemma}
\begin{proof}
The proof is divided into three parts. First we argue that $\tilde{T}_{\tilde{a}+1}$ could not have been left clean before the move $\f{\tilde{a}+2}{j}$. Then we argue that $\tilde{T}_{\tilde{a}+1}$ has to be clean before move $\f{\tilde{a}}{1}$.
Finally we analyze the construction of  $\tilde{T}_{\tilde{a}+1}$ to show that cleaning the subtrees of $\tilde{T}_{\tilde{a}+1}$ requires certain sets of searchers.
These sets of searchers are listed as a family $\cX$ and $\cY$.

Let us consider moves performed only in the attempt $\F$ which by \ref{eq:reversability-nm} cleans $\tilde{P}$ from $\tilde{v}_{\tilde{l}}$ to $\tilde{v}_1$.
Choose the minimal $j$ such that $\f{\tilde{a}+2}{j}\in \atts{\tilde{P}_{\tilde{a}}}{i(0)}{S} $. Hence, the edge $\{\tilde{v}_{\tilde{a}},\tilde{v}_{\tilde{a}+1}\}$ is not clean at move $\f{\tilde{a}+2}{j}$ move. 
The subtree $\tilde{T}_{\tilde{a}+1}$ cannot be completely clean in the move  $\f{\tilde{a}+2}{j}$, because it contains the vertex $\tilde{v}_{\tilde{a}+1}$, which is unoccupied (by the definition of $\f{\tilde{a}+2}{1}$) and incident to the contaminated edge $\{\tilde{v}_{\tilde{a}+1}, \tilde{v}_{\tilde{a}}\}$ (by \ref{eq:reversability-nm}). 

On the other hand each copy of $\tilde{L}_x$ has to be either completely clean or guarded in the move  $\f{\tilde{a}}{1}$. Suppose otherwise for a contradiction, then the contamination spreads unobstructed  through $\tilde{v}_{\tilde{a}+1}$, 
which cannot be occupied by a searcher during move $\f{\tilde{a}}{1}$,  to $\tilde{v}_{\tilde{l}}$ and, by the Lemma~\ref{lemma:fail}, the attempt $\F$ fails contrary to its definition.

Because the two searchers of color $\RED$ are not in a non-leaf vertex of $\tilde{L}_x$ in neither of the moves number $\f{\tilde{a}}{1}$ and $\f{\tilde{a}+2}{j}$ at most four copies of $\tilde{L}_x$ can be guarded at each of these moves. In total at most eight out of eleven copies of $\tilde{L}_x$ can be cleaned only partially between these two moves. Which means that in the attempt $\atts{\tilde{P}_{\tilde{a}}}{i(0)}{S}$ 
there exists a $\atts{\tilde{L}_{x,1}\cup \tilde{L}_{x,2}\cup \tilde{L}_{x,3}}{k}{S}$.

By construction, for each of set:
\[X\in\cX=\{\{\RED, \NPCvariable{1},\NPCvalve{1}{1},\NPCvalve{1}{2}, x\}\st x\in \{\NPCtrue{1}, \ldots, \NPCtrue{n}\}\cup\{\NPCfalse{1}, \ldots, \NPCfalse{n}\}\}\]
there exists  a  subtree $\tilde{L}_{x}$ requiring searchers of these colors. 
Note that $\sn{\tilde{L}_{x}}=5$ and $\sn{\tilde{L}_{x,1}\cup \tilde{L}_{x,2}\cup \tilde{L}_{x,3}}=6$, thus all searchers of colors contained in $X$ will be present on some vertices of $\tilde{L}_{x,1}\cup \tilde{L}_{x,2}\cup \tilde{L}_{x,3}$  simultaneously, in at least one move, whose number is contained in $\atts{\tilde{L}_{x,1}\cup \tilde{L}_{x,2}\cup \tilde{L}_{x,3}}{k}{S}$. 
Denote the number of the first such move in this attempt as $t_{X}$. Because we consider only moves whose numbers belong to $\F$, one searcher of color $\RED$ is present on $\tilde{P}$. In the move $t_{X}$ this searcher occupies $\tilde{v}_{\tilde{a}}$, therefore  $t_K\leq\f{\tilde{a}}{1}$.

The same argument can be repeated for any $\tilde{L'}_{y}$ and the respective set from $Y\in\cY=\{\{\RED, \NPCclause{i}\}\st i\in\{1,\ldots,m\}\}$ to prove existence of analogously defined $t_Y$. $\cK=\cX\cup \cY$ finishes the proof.
\end{proof}
We skip the proof of the following lemma because it is analogous to the one of Lemma \ref{lemma:mc}.
\begin {lemma}\label{lemma:mc2}
Let 
\[\cK'' = \{ \{\RED, \NPCclause{i}\}\st i\in\{1,\ldots,m\}\} . \]

For each $K''\in \cK''$, in $\atts{\tilde{P}_{\tilde{l}-2}}{i(0)}{S}$, there exists a move  $t_{K''}$ which requires all searchers of colors from the set $K''$ to be  to be in one of the subtrees $\tilde{L''}_{\NPCclause{d}}, d\in\{1,\ldots, m\}$.\qed
\end {lemma}

Let $\tilde{\cT}$ be the set of all subtrees $\tilde{S}_p, \tilde{S}_{-p}, \tilde{H}_{\NPCvalve{p}{1}}(\tilde{R}), \tilde{H}_{\NPCvalve{p}{2}}(\tilde{R}), p\in\{1,\ldots, n\}$.
Recall that these subtrees are attached to the vertices $\tilde{v}_j$ for $j\in \tilde{R}$.

\begin {lemma}\label{lemma:allcleanm} 
In the move  $\f{\tilde{a}}{1}$ all subtrees in $\tilde{\cT}$ are clean.
\end {lemma}
\begin{proof}
Each move $t_K$ introduced in Lemma~\ref{lemma:mc}, where $K\in \cK$, happens before the vertex $v_a$ is reached in the first successful attempt to clean $\tilde{P}$. 
Each subtree in $\tilde{\cT}$ contains only vertices of colors found in some $K\in \cK$.
Every subtree in $\tilde{\cT}$ is connected to vertices which were cleaned before $v_a$.
Because in the move $t_K$ all searchers of colors present in some subtree of $\tilde{\cT}$ are in $T_{a+1}$ if this subtree  of $\tilde{\cT}$ contains a contaminated edge, then the attempt to clean $\tilde{P}$ fails --- since we analyze a successful attempt, a contradiction occurs.
Finally we show that a recontamination of this subtree  of $\tilde{\cT}$  cannot happen prior to the move $\f{\tilde{a}}{1}$.

By Lemma~\ref{lemma:mc}, $\f{a}{1}\geq t_K$ for each $K\in \cK$.
By construction, for each subtree $\tilde{G}\in\tilde{\cT}$ there exists $K\in \cK$ such that the set of colors in vertices of $\tilde{G}$, denoted by $\cset{\tilde{G}}$, is a subset of $K$.
Note that any subtree in $\tilde{\cT}$ is connected to the vertex $\tilde{v}_{\tilde{a}+1}$ only by a subpath of $\tilde{P}$, which may contain a subset of the following vertices $\{\tilde{v}_{\tilde{a}+1},\ldots, \tilde{v}_{\tilde{b}}\}$, and all these vertices are of color $\RED$.  %Because both of the searchers of color $\RED$ are in $T_{a+1}$ (one explicitly on the vertex $u_{a+1}$), then if $G$ contained    

Suppose for a contradiction that $\tilde{G}$ contains a contaminated edge in the move $t_K$ such that $\cset{\tilde{G}}\subseteq K$. Because all searchers in colors $\cset{\tilde{G}}$ are in $\tilde{T}_{\tilde{a}+1}$ (one of color $\RED$ is explicitly on the vertex $\tilde{v}_{\tilde{a}+1}$) and $\tilde{G}\cap \tilde{T}_{\tilde{a}+1}= \emptyset$, $\tilde{G}$ contains no searchers in the move $t_K$. If it were to contain a contaminated edge at this point, then the contamination would have spread unobstructed along the path $\tilde{P}$ from a vertex by which $\tilde{G}$ is attached (which may be one of the following $\{\tilde{v}_{\tilde{a}+2},\ldots, \tilde{v}_{\tilde{b}}\}$) to the edge $\{\tilde{v}_{\tilde{l}}, \tilde{v}_{\tilde{l}-1}\}$. By Lemma~\ref{lemma:fail}, the attempt $\F$ fails, which contradicts its definition.   

If $\tilde{G}$ contained no contaminated edge nor searchers and was adjacent to the $\clean{\tilde{P}}{t_K)}$, then in the move  $t_K$ it was completely clean. By construction, recontamination may be introduced to $\tilde{G}$ only through the vertex of $\tilde{P}$ by which it is attached. Because $t_K\geq \f{\tilde{a}+2}{j}$ there is always a searcher of color $\RED$ guarding it from   $\cont{\tilde{P}}{t}, t\in[t_K, \f{\tilde{a}}{1}]$, so it stays clean.
\end{proof}

\begin {lemma}\label{lemma:oneclean}
There is at least one clean edge in each $\tilde{A}_r, r\in \tilde{R}$, in each move of $\F$.
\end {lemma}
\begin{proof}
Assume for the sake of a contradiction, that an area  $\tilde{A}_r$ is fully contaminated in the move  $\f{\tilde{l}}{1}$. $\sn{\tilde{A}_r}=2$ so it cannot be cleaned in $\F$, because at least one searcher of color $\RED$  is in $\tilde{P}$. This contradicts Lemma~\ref{lemma:allcleanm}, because
the move $\f{\tilde{a}}{1}$, in which all of the subtrees in $\tilde{\cT}$
are clean, belongs to $\F$.  
\end{proof}

Let $\tilde{P}_{\tilde{b}}^{+}$ denote $\induced{\TSATP}{\{\tilde{v}_{\tilde{l}},\ldots, \tilde{v}_{\tilde{b}}, v\}}$ where $v$ is a leaf incident $\tilde{v}_{\tilde{b}}$ which does not belong to $\tilde{P}$.
Consider the attempt $\Fb$.
Note that  $\f{\tilde{l}}{1}\in\Fb,\f{\tilde{b}}{1}\in \Fb$.
Informally speaking, we define the first successful attempt to clean the subtree containing clause components, which are by construction connected to the vertices of the path $\tilde{P}_{\tilde{b}}$.

\begin {lemma}\label{lemma:atleastoneAll}
There is at least one searcher in each $\tilde{G}\in\tilde{\cT}$ guarding each $\clean{\tilde{A}_r}{t}$ where  $r\in \tilde{R}$, $t\in \Fb$.
\end {lemma}
\begin{proof}
Because the move $t$ belongs to $\F$, by Lemma~\ref{lemma:oneclean} there is at least one clean edge in $\tilde{A}_r$ which has to be separated from $\cont{\tilde{P}}{t}$. By the definition of $t$, at least one searcher of color $\RED$ is on $\tilde{v}_c, c>\tilde{b}$ and edges $\{\{\tilde{v}_{\tilde{b}}, \tilde{v}_{\tilde{b}+1}\},\ldots, \{\tilde{v}_{\tilde{l}}, \tilde{v}_{\tilde{l}-1}\}\}$ are contaminated. Therefore there is at least one contaminated edge incident to each $G$, and a searcher guarding $\clean{\tilde{A}_r}{\f{\tilde{b}}{k}}$ can only be placed in  the subtree $\tilde{G}\in \tilde{\cT}$ containing $\tilde{A}_r$. 
\end{proof}

Note that in $\TSATP$ the areas $\tilde{A}_r, r\in\tilde{R}$, are subgraphs of $\tilde{H}_{o_{p}}(\tilde{R})$, $\tilde{S}_p$ and $\tilde{S}_{-p}$ (while only ${S}_p$ and ${S}_{-p}$ were included in $\TSAT$), hence the two separate lemmas below.
The lemmas state that in the moves in which two searchers of color $\RED$ are on the path $P$  within the attempt to clean the clause components, only searchers of  colors different than $\RED$ can be used to guard clean edges of the subtrees $\tilde{H}_{o_{p}}(\tilde{R})$, $\tilde{S}_p$ and $\tilde{S}_{-p}$.

\begin {lemma}\label{lemma:atleastoneO}
For any $i$ and $j$, such that $\f{i}{j}\in\Fb$, in a move  $\f{i}{j}$ there is a searcher of color $o_{p}\in\{\NPCvalve{p}{1}, \NPCvalve{p}{2}\}, p\in\{1,\ldots,n\}$, in  $\tilde{H}_{o_{p}}(\tilde{R})$. 
\end {lemma}
\begin{proof}
By the definition of $\f{i}{j}$, there can be no searcher of color $\RED$ on any vertex of $\tilde{H}_{o_{p}}(\tilde{R})$. By Lemma \ref{lemma:atleastoneAll},   each of them contains a searcher, and by colors of vertices in $\tilde{H}_{o_{p}}(\tilde{R})$, it is of color $o_{p}$. 
\end{proof}

Note that Lemmas~\ref{lemma:atleastoneO} and \ref{lemma:atleastoneV} speak of the same moves, therefore the pool of available searchers is shared between them.

\begin {lemma}\label{lemma:atleastoneV}
For any $i$ and $j$, such that $\f{i}{j}\in\Fb$, in a move  $\f{i}{j}$ there is a searcher of color $\NPCtrue{p}$ (respectively $\NPCfalse{p}$) or $\NPCvariable{p}$ in  $\tilde{S}_p$ ($\tilde{S}_{-p}$ respectively).% in $\f{i}{j}\in\Fb$.
\end {lemma}
\begin{proof}
By the definition of $\f{i}{j}$, there can be no searcher of color $\RED$ on any vertex of $\tilde{S}_p$. By Lemma \ref{lemma:atleastoneAll},   each of them contains a searcher, and by colors of vertices in $\tilde{S}_p$ and Lemma \ref{lemma:atleastoneO}, it is of color $\NPCtrue{p}$ or $\NPCvariable{p}$. Proof for $\tilde{S}_{-p}$ is analogical.
\end{proof}

\subsection{Adaptation to non-monotonicity --- there is no going back}

Because of a possibility of recontamination, the previous lemmas are insufficient to obtain a result analogous to that given by Lemma~\ref{lemma:guard}. In this section we find a configuration of searchers that cannot be used in $\tilde{P}_{\tilde{b}}^{+}$ in a successful attempt to clean $\tilde{P}_{\tilde{b}}^{+}$, cf. Lemma~\ref{lemma:summary}.

\begin {lemma}\label{lemma:switch}
At least one of searchers of color from the following set: $Q=\{ \NPCvalve{p}{1},\NPCvalve{p}{2}, x_{p}\}, x_{p}\in \{\NPCtrue{p}, \NPCfalse{p}\}$, has to remain in each $\tilde{S}_{p}\cup\tilde{S}_{-p}$ in each move  $t\in[\f{\tilde{l}}{1},\f{\tilde{b}}{1}]\subseteq \Fb$.
\end {lemma}

\begin{proof}
The proof revolves around analyzing colors of vertices in $\tilde{S}_p$ and $\tilde{S}_{-p}$ which can be used by guarding searchers in Lemma~\ref{lemma:atleastoneAll}.
A switch is a change in the guarding searchers of colors different than $\NPCvalve{p}{1}$ or  $\NPCvalve{p}{2}$ in $\tilde{S}_p$.
In order to make such a switch, we either have to clean $O_p$ or recontaminate it.
We exclude the possibility to clean any star $O_{p}$ in the moves $t$ listed in the lemma (see Observation~\ref{obs}).
Thus, we have to consider recontamination of $O_p$.
Then, we use Lemma~\ref{lemma:atleastoneV} to establish that such a switch can occur once in $\tilde{S}_p$ (or $\tilde{S}_{-p}$ analogically). Finally we look at  guarding searchers in both  $\tilde{S}_p$ and $\tilde{S}_{-p}$ in the moves $\f{\tilde{l}}{1}$ and $\f{\tilde{b}}{1}$  and show that a switch can occur in either $\tilde{S}_p$ or $\tilde{S}_{-p}$, but not both, between these moves.

Pick a vertex, denoted by $u_{p, t}$, such that $u_{p, t}\in V(\tilde{S}_p$) ($u_{-p, t}\in V(\tilde{S}_{-p})$ respectively) and it is incident to a contaminated and a clean edge in the move  $t\in \Fb$.
Let us focus only on $u_{p,t}$ as the approach is analogous for  $u_{-p,t}$.
By Lemma \ref{lemma:atleastoneAll},  this vertex exists.

Denote the set of colors other than $\NPCvalve{p}{1},\NPCvalve{p}{2}$ in the set of colors of $u_{p, t}$ as $c_{p,t}$, i.e., $c_{p,t}=\colV(u_{p,t})\setminus\{\NPCvalve{p}{1},\NPCvalve{p}{2}\}$.  
Recall that $\tilde{S}_p$ and $\tilde{S}_{-p}$ contain copies of the star $O_{p}$. 
Note that if $c_{p,t}=\emptyset$ then  $u_{p, t}$ is a central vertex of such a copy of  $O_{p}$. Otherwise $|c_{p,t}|=1$. Let $f^{+}(t)$ denote the number of the first move such that $f^{+}(t)\geq t$ and $c_{p,f^{+}(t)}\neq\emptyset$.
Let $f^{-}(t)$ denote the number of the first move such that $f^{-}(t)\leq t$ and $c_{p,f^{-}(t)}\neq\emptyset$. 

\begin{observation}\label{obs}
By Lemma \ref{lemma:atleastoneAll}, no $\atts{O_{p}}{i}{S}\subseteq[\f{\tilde{l}}{1},\f{\tilde{b}}{1}] $ exists.
\end{observation}

By construction, any maximal subtree $\tilde{T}$, such that $\tilde{T}$ is a subgraph of $\tilde{S}_p$ where $u_{p,f^{-}(t)}$ and $u_{p,f^{+}(t')}, t<t'$, are leaves and $c_{p, f^{-}(t)}\neq c_{p,f^{+}(t')}$, contains a copy of $O_{p}$, to which we refer further as $O'_{p}$. If $c_{p, h}=\emptyset$, then $h\in \att{O_{p}}{i}$ and $f^{-}(h)$ corresponds to the beginning of this attempt ($f^{+}(h)$ corresponds to its end, respectively).
Note that both $c_{p,\f{\tilde{l}}{1}}\neq \emptyset$ and $c_{p,\f{\tilde{b}}{k}}\neq \emptyset$.  
By Observation~\ref{obs}, a pair $u_{p,f^{-}(t)}$ and $u_{p,f^{+}(t')}$, such that $t, t'\in [\f{\tilde{l}}{1},\f{\tilde{b}}{k}]$, which satisfies  $c_{p, f^{-}(t)}\neq c_{p,f^{+}(t')}$,  exists only if $O'_{p}\subseteq \tilde{T}$ has already been clean in the move  $f^{-}(t)$. 
Similar argument can be repeated for a pair $u_{p,t}$ and $u_{-p,t'}$ (i.e. vertices in $\tilde{S}_{-p}$ and $\tilde{S}_{-p}$) with the conclusion that a pair which satisfies the above constraints does not exist --- there are no clean copies of $O_p$ between them.
Informally, we can switch a searcher of color in $c_{p, t}\neq \emptyset$ to a searcher of color in $c_{p, t'}\neq c_{p, t}\neq \emptyset$ only by causing recontamination, and we cannot use the first searcher again. Thus, there is a finite number of switches in $\Fb$.

Note that by Lemma~\ref{lemma:atleastoneV}, $c_{p,\f{l}{1}}$ ($c_{-p,\f{l}{1}}$ respectively) contains either  $\NPCtrue{p}$ (respectively $\NPCfalse{p}$)  or $\NPCvariable{p}$. If $c_{p, t'}=\{\RED\}$  we contradict Lemma~\ref{lemma:atleastoneV} in the next move  $\f{i}{j}$ after $t'$.  Therefore, a set $c_{p, t}$ or $c_{p, t'}$ can contain only a one out of these two colors: $\NPCtrue{p}$,  $\NPCvariable{p}$, or be empty.
By the previous paragraph and the number of different colors, there exists at most one interval of numbers of moves 
$J=[f^{-}(j),f^{+}(j')]$ such that $c_{p, f^{-}(j)}=\NPCtrue{p}$ and $c_{p, f^{+}(j')}=\NPCvariable{p}$ or vice versa. Informally, we can switch the color of required searcher once.
The same argument holds for $\tilde{S}_{-p}$ and colors $\NPCfalse{p}$,  $\NPCvariable{p}$ Denote the corresponding interval as $L$.

Because there is only one $\NPCvariable{p}$ searcher at least one of the following is true: $c_{p,\f{l}{1}}=\{\NPCtrue{p}\}$ or $c_{-p,\f{l}{1}}=\{\NPCfalse{p}$\}, thus at most one edge of color $\NPCvariable{p}$ in $\tilde{S}_{-p}\cup \tilde{S}_{p}$  is clean. 
For the same reason at most one the following is true: $c_{p,\f{b}{k}}=\{\NPCvariable{p}\}$ or $c_{-p,\f{b}{k}}=\{\NPCvariable{p}\}$. Thus, in a single strategy at most one of the two intervals $J$ and $L$ exists. If $J$ ($L$ respectively) does not exist, then $c_{p,t}\in\{c_{p,\f{l}{1}}, \emptyset\}$ ($c_{-p,t}\in\{c_{-p,\f{l}{1}}, \emptyset\}$ respectively) for each move $t\in\Fb$. By definition of $u_{p,t}$, only searchers of colors $\NPCvalve{p}{1},\NPCvalve{p}{2}$ and those in $c_{p,t}$ can stay in $\tilde{S}_p$ in each move of $[\f{\tilde{l}}{1},\f{\tilde{b}}{1}]$.
\end{proof}

\begin{lemma}\label{lemma:looseend}
If the searcher of color  $\NPCvalve{p}{z}, z\in \{1,2\}$ is not in in $\tilde{H}_{\NPCvalve{p}{z}}(\tilde{R})$ in a move  $t\in\Fb$, then a searcher of color $\RED$ is in $\tilde{H}_{\NPCvalve{p}{z}}(\tilde{R})$.
\end{lemma}
\begin{proof}
Because $t\in\F$  and by Lemma~\ref{lemma:oneclean} at least one searcher has to remain in $\tilde{H}_{\NPCvalve{p}{z}}(\tilde{R})$. It can be of color $\RED$ or $\tilde{H}_{\NPCvalve{p}{z}}(\tilde{R})$.
\end{proof}

\begin{lemma}\label{lemma:summary}
There exists a set of searchers of colors $\{ \RED, \NPCvalve{p}{1},\NPCvalve{p}{2}, x_{p}\st p\in \{1,\ldots,n\}\}, x_{p} \in \{\NPCtrue{p}, \NPCfalse{p}\}$ such that all but one have to remain outside of $\tilde{P}_{\tilde{b}}^{+}$ in each move  $t\in[\f{\tilde{l}}{1},\f{\tilde{b}}{1}]\subseteq \Fb$.
\end{lemma}
\begin{proof}
By Lemma~\ref{lemma:switch}, for each $p\in\{1,\ldots,n\}$ at least one of searchers of color from the following set: $Q=\{ \NPCvalve{p}{1},\NPCvalve{p}{2}, x_{p}\}, x_{p}\in \{\NPCtrue{p}, \NPCfalse{p}\}$, has to remain in $\tilde{S}_{p}\cup\tilde{S}_{-p}$ in each move  $t\in[\f{\tilde{l}}{1},\f{\tilde{b}}{1}]$. Let $s$ denote such a searcher of color other than $x_{p}$. If $s$  exists then by Lemma~\ref{lemma:looseend} a searcher of color $\RED$ is in $\tilde{H}_{\NPCvalve{p}{z}}(\tilde{R})$. In a move $t\in[\f{\tilde{l}}{1},\f{\tilde{b}}{1}]$ at least one searcher of color $\RED$ is on the path $\tilde{P}_{\tilde{b}}$ so $s$ is unique.
Then, $(\tilde{S}_{p}\cup\tilde{S}_{-p})\cap\tilde{P}_{\tilde{b}}^{+} =\emptyset $  and $\tilde{H}_{\NPCvalve{p}{z}}(\tilde{R})\cap\tilde{P}_{\tilde{b}}^{+} =\emptyset $ finish the proof.
\end{proof}

To informally summarize, we show that there exists a set of searchers of colors $\{ \RED, \NPCvalve{p}{1},\NPCvalve{p}{2}, x_{p}\st p\in \{1,\ldots,n\}\}, x_{p} \in \{\NPCtrue{p}, \NPCfalse{p}\}$
 of which at most one at a time can take part in cleaning of $\tilde{P}_{\tilde{b}}^{+}$.

\subsection{Conclusion}

\begin{lemma}\label{lemma:3SATnm}
Let $x_{1},\ldots, x_{n}$ and Boolean formula $C=C_{1}\land C_{2}\ldots\land C_{m}$ be an input to the $\problemSAT$ problem.
If there exists a search strategy using $2+5n+2m$ searchers for $\TSATP$, then the answer to $\problemSAT$ problem is $\YES$.
\end{lemma}
\begin{proof}
The proof revolves around the configuration of searchers in the move  $\f{\tilde{b}}{i(0)}$.
We define a Boolean assignment as follows: $x_{p}$ is true if and only if a searcher of color $\NPCtrue{p}$  does \emph{not} guard the area $A_{\tilde{a}+1+p}$ in the move  $\f{\tilde{b}}{i(0)}$, otherwise $x_p$ is false. Let $X\subset \{\NPCtrue{i}, \NPCfalse{i}\st i\in\{1,\ldots,n\}$ denote the colors of those searchers.  By Lemma~\ref{lemma:summary}, a valid assignment will occur during execution of $\tilde{c}$-search strategy using $2+5n+2m$ searchers. We omit the detailed proof in favor of an analogy to the proof of Lemma~\ref{lemma:3SAT}. 

Let $\tilde{T}$ denote the maximal subtree containing  $\tilde{v}_{\tilde{l}}$ such that $\tilde{v}_{\tilde{b}}$ is this subtree's leaf. Note that $\tilde{T}$ is  isomorphic to its equivalent in $\TSAT$, and monotone strategies are a subset of strategies available in this version of the problem. We focus on proving that the configuration of searchers in the move  $\f{\tilde{b}}{1}$ has properties analogous to those of the configuration in the move $\first{b}$ of a strategy for $\TSAT$.
Regardless of the moves performed by searchers in a $\tilde{c}$-strategy if a color $x\notin X$, then an edge of color $x$ in $E(\tilde{T})$ which was contaminated before the move  $\f{\tilde{l}}{1}$ remains contaminated in the moves of numbers from the interval $[\f{\tilde{l}}{1}, \f{\tilde{b}}{1}]$. Thus, configurations which do not correspond to a valid assignment cannot use the searchers of appropriate colors required to guard them and continue cleaning the tree.

All that remains to be addressed is the possibility of these edges being clean before the move  $\f{\tilde{l}}{1}$ (recall that in the proof of Lemma~\ref{lemma:3SAT} we used the notion of monotonicity to resolve this issue, here the argument has to be continued).  If this was the case they would have to be guarded by at least one searcher in the moves of numbers from the interval $[\f{\tilde{l}}{1}, \f{\tilde{h}}{1}]$. By Lemma~\ref{lemma:mc2}, there exists a move whose number is in  this interval such that all searchers of color $\RED$ and $\NPCclause{d}$ are not in $\tilde{L}_{d, 1}$, $\tilde{L}_{d, 2}$, $\tilde{L}_{d, 3}$. Thus, by the colors of vertices of $\tilde{L}_{d, i}$ only  the searcher of color $x$ can prevent recontamination of an edge of color $x$ and it cannot be used, by the definition of $X$.  Furthermore, these edges stay contaminated in the move  $\f{\tilde{b}}{i(0)}$.

 We use only the positions of searchers in a move of a specific number, so we are interested in a result, not the process, of a partial cleaning of $\TSAT$ and $\TSATP$. Therefore, most arguments from Lemma~\ref{lemma:3SAT} can be applied to $\TSATP$.
Recall the conclusion of the proof of Lemma~\ref{lemma:3SAT}. In order for a $\tilde{c}$-strategy for $\TSAT$ to exist at least one subtree $L_{d, 1}$, $L_{d, 2}$, $L_{d, 3}$ for each $d\in \{1,\ldots, m\}$ is cleaned before the clean part of $A_0$ reaches $v_b$, or all three have to be guarded, and for that to happen they have to contain edges in at least one color  corresponding to $\literal{d}{i}$ in clause $C_{d}$. The same is true for  a $\tilde{c}$-strategy for $\TSATP$ and  its respective counterparts of $\TSAT$.
\end{proof}

\section{Polynomially tractable instances} \label{sec:easy}

If $G$ is a tree then Lemma~\ref{lem:lower} gives us a~lower bound of $\lb{G}$ on the number of searchers. 
In this section we will look for an upper bound assuming that there is exactly one connected component per color.
With this assumption we show a constructive, polynomial-time algorithm both for {\problemHGS} and {\problemHCGS}.

\medskip
Let $(E_1,\ldots,E_k)$ be the partition of edges of $T$ so that $E_i$ induces the area of color $i$ in $T$.
Observe that this partition induces a~tree structure.
More formally, consider a graph in which the set of vertices is $P_E = \{E_1, E_2, \ldots E_k\}$ and $\{E_i, E_j\}$ is an edge if and only if an edge in $E_i$ and and edge in $E_j$ share a common junction in $T$.
Then, let $\SimpTree$ be the BFS spanning tree with the root $E_1$ in this graph.
We write $V_i$ to denote all vertices of the area with edge set $E_i$, $i\in\{1,\ldots,\cnum\}$.

Our strategy for cleaning $T$ is recursive, starting with the root.
The following procedure requires that when it is called, the area that corresponds to the parent of $E_i$ in $\SimpTree$ has been cleaned, and if $i\neq 1$ (i.e., $E_i$ is not the root of $\SimpTree$), then assuming that $E_j$ is the parent of $E_i$ in $\SimpTree$, a searcher of color $j$ is present on the junction in $V_i\cap V_j$.
With this assumption, the procedure recursively cleans the subtree of $\SimpTree$ rooted in $E_i$.
\begin{algorithmic}
\Procedure{\ProcClean}{labeled tree $T$, $E_i$}
\Comment{Clean the subtree of $T$ that corresponds to the subtree of $\SimpTree$ rooted in $E_i$}
\begin{enumerate}
 \item For each $E_j$ such that $E_j$ is a~child of $E_i$ in $\SimpTree$ 
		place a searcher of color $j$ on the~junction $v \in V_j \cap V_i$.  
\item \label{s:opt}
Clean the area of color $i$ using $\sn{T[V_i]}$ searchers. Remove all searchers of color $i$ from vertices in $V_i$.
\item For each child $E_j$ of $E_i$ in $\SimpTree$: 
	\begin{enumerate} 
	 \item place a searcher of color $i$ on the junction $v \in V_j \cap V_i$,
		\item \label{move:remove1} remove the searcher of color $j$ from the vertex $v$,
		\item call $\ProcClean$ recursively with input $T$ and $E_j$,
		\item \label{move:remove2} remove the searcher of color $i$ from the vertex $v$.
	\end{enumerate}
\end{enumerate}
\EndProcedure	
\end{algorithmic}

\begin{lemma}
\label{lem:upper}
For a~given tree $G = (V(G), E(G), \colE)$, procedure~{\ProcClean}($G$, $E_1$) cleans $G$ using 
$\lb{G}$ searchers. \end{lemma}

\begin{proof}
First, observe that the number of searchers used during the execution of procedure {\ProcClean} is exactly as specified.
Indeed, to clean each of the $T[V_i]$ we use $\sn{T[V_i]}$ searchers of color $i$ and at most one searcher of other colors.

Note that moves~\ref{move:removing} do not cause recontamination. 
Indeed, the move defined in step~\ref{move:remove1} of procedure~{\ProcClean} removes a~searcher from the node on which another searcher is present, 
while the move~\ref{move:remove2} is performed on node $v$ when both subtree and the parent subtree of $v$ are cleaned.
This gives the correctness of search strategy produced by procedure~{\ProcClean}.
\end{proof}

We also immediately obtain.
\begin{lemma}
If all the strategies used in step~\ref{s:opt} of procedure {\ProcClean} to clean a subtree $T[V_i]$ are monotone, then the resulting ${\colS}$-strategy for $G$ is also monotone.
\qed
\end{lemma}

It is known that there exists an optimal monotone search strategy for any graph~\cite{LaPaugh93} and it can be computed in linear time for a tree~\cite{MegiddoHGJP88}. An optimal connected search strategy can be also computed in linear time for a tree~\cite{BFFFNST12}.

Using Lemmas~\ref{lem:lower} and \ref{lem:upper} we conclude with the following theorem:
\begin{theorem}
Let $G = (V(G), E(G), \colE)$ be a tree such that the subgraph $G_j$ composed by the edges in $E_j$ is connected for each $j \in \{ 1, 2, \ldots, \cnum\}$. Then, 
there exists a~polynomial-time algorithm for solving problems $\problemHGS$ and $\problemHCGS$. 
\end{theorem}

\section{Conclusions and open problems} \label{sec:conclusions}

Recalling our main motivation standing behind introducing this graph searching model, we note that its properties allow for much easier construction of graphs in which recontaminations need to occur in optimal strategies.
Our main open question, following the same unresolved one for connected searching, is whether problems $\problemHGS$ and $\problemHCGS$ belong to NP?

Our more practical motivation for studying the problems is derived from modeling physical environments to whose parts different robots have different access.
More complex scenarios than the one considered in this work are those in which either an edge can have multiple colors (allowing it to be traversed by all searchers of those colors), and/or a searcher can have multiple colors, which in turns extends its range of accessible parts of the graph.
As a way of modeling mobile entities of different types cooperating to solve various computational tasks (of which searching is just one example), heterogeneous mobile agent computing is receiving a growing interest, including fields like distributed computing. Hence, one may ask for different ways of modeling differences between searchers, which may fit potential practical applications.

\bibliographystyle{plain}
\bibliography{references}

\end{document}